\documentclass[11pt]{article}
\usepackage{parskip}
\usepackage{amsmath}
\usepackage{amsfonts}
\usepackage{amsthm}
\usepackage{bbm}
 \usepackage{threeparttable}
\usepackage{booktabs}
\usepackage{comment}
\usepackage{amssymb}
\usepackage{geometry}
\usepackage{stackrel}
\usepackage{bm}
\usepackage{xcolor}
\usepackage{graphicx}
\usepackage[onehalfspacing]{setspace}
\usepackage[authoryear,round,longnamesfirst]{natbib}
\usepackage{grffile}
\usepackage{bbm}
\usepackage[normalem]{ulem}
\usepackage{mathtools}
\usepackage{etoolbox}

\usepackage[colorlinks=true, pdfstartview=FitV, linkcolor=blue, citecolor=blue, urlcolor=blue]{hyperref}

\newtheorem{assumption}{Assumption}
\newtheorem{theorem}{Theorem}

\newtheorem{lemma}{Lemma}
\newtheorem{proposition}{Proposition}
\newtheorem{remark}{Remark}

\newcommand{\idf}{\mathbbm{1}}
\newcommand{\norm}[1]{\left\lVert #1 \right\rVert}
\newcommand{\E}{\mathbbm{E}}
\newcommand{\var}{\text{\upshape{Var}}}
\newcommand{\cov}{\text{\upshape{Cov}}}
\newcommand{\longto}{\longrightarrow}
\newcommand{\pto}{\overset{p}{\rightarrow}}
\newcommand{\longpto}{\overset{p}{\longrightarrow}}

\newcommand{\longdto}{\overset{d}{\longrightarrow}}
\newcommand{\numberthis}{\addtocounter{equation}{1}\tag{\theequation}}

\newcommand{\R}{\mathbb{R}}

\newcommand{\N}{\mathbb{N}}
\newcommand{\dd}{\,d}
\newcommand{\vpd}[2]{\ensuremath{\frac{\partial #1}{\partial #2}}}

\newcommand{\given}[1][]{\,#1|\,}
\usepackage[hang,flushmargin]{footmisc}

\title{Panel Quantile Regression with Common Shocks\thanks{First arXiv version: 20 February 2026. We are grateful to Bruce E. Hansen and participants at NY Camp Econometrics XX for helpful comments and suggestions. All the remaining errors are ours.\medskip}}
\author{
	Harold D. Chiang\thanks{Harold D. Chiang: hdchiang@wisc.edu. Department of Economics, University of Wisconsin-Madison, William H. Sewell Social Science Building, 1180 Observatory Drive,	Madison, WI 53706, USA\medskip} 
	\quad 
	Antonio F. Galvao\thanks{	Antonio F. Galvao:  agalvao@msu.edu. 
Department of Economics,
Michigan State University, 
110 Marshall-Adams Hall, 
486 W. Circle Drive 
East Lansing, MI 48824, USA\medskip} 
	\quad 
Chia-Min Wei\thanks{Chia-Min Wei: cwei69@wisc.edu. Department of Economics, University of Wisconsin-Madison, William H. Sewell Social Science Building, 1180 Observatory Drive,	Madison, WI 53706, USA\medskip} 
}

\begin{document}
\maketitle

\begin{abstract}
This paper develops an asymptotic and inferential theory for fixed-effects panel quantile regression (FEQR) that delivers inference robust to pervasive common shocks. Such shocks induce cross-sectional dependence that is central in many economic and financial panels but largely ignored in existing FEQR theory, which typically assumes cross-sectional independence and requires $T \gg N$. We show that the standard FEQR estimator remains asymptotically normal under the mild condition $(\log N)^2/T \to 0$, thereby accommodating empirically relevant regimes, including those with $T \ll N$. We further show that common shocks fundamentally alter the asymptotic covariance structure, rendering conventional covariance estimators inconsistent, and we propose a simple covariance estimator that remains consistent both in the presence and absence of common shocks. The proposed procedure therefore provides valid robust inference without requiring prior knowledge of the dependence structure, substantially expanding the applicability of FEQR methods in realistic panel data settings.

\end{abstract}

\vspace{0.5cm}


\medskip{}

\section{Introduction}

Quantile regression \citep{koenker1978regression} offers a flexible framework for analysing heterogeneity across the conditional distribution of an outcome and is particularly valuable for studying tail behaviour and outcomes with heavy tails or outliers. In many economic and financial applications, data are naturally organised as panels, making panel quantile regression with individual fixed effects (FEQR) an important tool for modelling distributional heterogeneity while controlling for unobserved unit-specific heterogeneity (\citet{Koenker04}). A salient empirical characteristic of such panels, however, is the prevalence of strong common shocks.

Common shocks—aggregate disturbances or latent factors that affect many cross-sectional units within a given period—are a central and recurring feature of economic and financial data. Much empirical work is built around this structure; for example, asset-pricing frameworks such as the two-pass procedure emphasise time variation in risk premia driven by economy-wide conditions rather than purely idiosyncratic firm shocks \citep{fama1973risk}. More broadly, macroeconomic announcements, monetary policy changes, business-cycle fluctuations, regulatory reforms, and shifts in global risk sentiment generate substantial comovement across units, so regression disturbances typically display cross-sectional dependence even after conditioning on observables. As emphasised in \cite{andrews2003cross}
, this pervasiveness reflects the fact that key determinants of behaviour—interest rates, inflation, financial conditions, policy actions, institutional and legal changes, political events, environmental and health shocks, and technological change—operate at an aggregate level, simultaneously influencing wages, consumption, investment, costs, wealth, and credit conditions for large sets of agents. Although exposure varies with sector, location, or balance-sheet characteristics, these forces link units through shared systemic influences, and with deepening economic integration their transmission has become more widespread. Hence, cross-sectional independence is rarely a credible assumption in practice; common shocks are a structural feature of economic environments rather than an exceptional complication.

These features have first-order econometric consequences. In cross-sectional settings, \cite{andrews2005cross} shows that least squares can be inconsistent when an unobserved common shock  affects both regressors and errors, inducing systematic dependence between them. Intuitively, the key requirement for consistency--that regressors are uncorrelated with the error term--is violated once this dependence varies with the realisation of the common shock, so that averaging across individuals no longer eliminates the bias. An analogous issue arises in quantile regression, where the  score object replaces the  product of regressor and error; if this object also depends on the common shock in a systematic way, similar inconsistency emerges. In panel data with common shocks, the error process violates cross-sectional independence and often lies outside the weak-dependence frameworks underlying classical panel asymptotics. \citet{driscoll1998consistent} show that conventional covariance estimators are generally inconsistent under such dependence and propose standard errors that are robust to broad forms of cross-sectional and temporal correlation. In empirical finance settings, \citet{petersen2008estimating} demonstrate that ignoring common time effects leads to substantial size distortions, with severely overstated t-statistics. 
This concern is especially acute in panels with fairly large $N$ but moderately large $T$, where limited time variation restricts information about aggregate components.  For more recent progress in linear panel data with common shocks, see e.g. \cite{chiang2024standard,juodis2025shock}.

Against this background, the theoretical literature on FEQR has progressed considerably. \cite{Koenker04}  first  studies the estimation problem of a linear quantile regression model with individual fixed effects, proposing an $\ell_1$-regularised estimator to mitigate the incidental parameter problem and analysing its asymptotic properties.
\citet{canay2011simple} propose a two-step estimator that eliminates fixed effects under a location-shift restriction, delivering standard large-$N$, large-$T$ asymptotics. Without imposing a location shift, \citet{kato2012asymptotics} develop a general asymptotic theory for Koenker-type FEQR with individual specific effects, establishing consistency and asymptotic normality under joint growth of $N$ and $T$ subject to a long-panel rate condition such as $N^2(\log N)^3/T \to 0$, and highlighting the challenges posed by incidental parameters and the non-smooth objective. \citet{galvao2016smoothed} introduce a kernel-smoothed version of the quantile objective to enable higher-order expansions, characterise the incidental-parameter bias, and propose analytic bias correction. More recently, \citet{galvao2020unbiased} obtain asymptotically unbiased normal approximations under a weaker long-panel requirement, $N(\log T)^2/T \to 0$, which is standard in the nonlinear panel data models literature (see, e.g., \citet{ArellanoHahn07}, \citet{ArellanoBonhomme11}, and \citet{FernandezValWeidner18}). Nevertheless, these results still rely on regimes in which the time dimension dominates the cross section, effectively requiring $T \gg N$, which limits applicability in many economic and financial panels. Other recent approaches include \citet{MachadoSantosSilva19} considering a location-scale shift effect for the individual effects, while adopting a location-scale shift specification also for the slope parameter, and \citet{GuVolgushev19} and \citet{ZhangWangLianLi2023} investigating the estimation of panel quantile regression models with group effects.\footnote{The panel quantile literature offers different identification and estimation strategies including penalized fixed effects, random effects, correlated random effects, group effects, instrumental variables, and factor models (see, e.g., \citet{Koenker04}; \citet{AbrevayaDahl08}; \citet{Lamarche10}; \citet{KimYang11}; \citet{ChetverikovLarsenPalmer16}; \citet{ArellanoBonhomme16}; \citet{GrahamHahnPoirierPowell2015};
\citet{demetrescu2023tests}).}  

Despite its increasing empirical relevance, the theory of fixed-effects panel quantile regression still exhibits two notable gaps. First, most existing results abstract from common shocks and therefore do not account for the cross-sectional dependence that is inherent in many economic and financial panels. In the linear regression context, theoretical analysis of common time effects in cross-sectional regressions was developed by \cite{andrews2005cross}. Comparable results, however, appear to be unavailable for panel quantile regression models.
 Second, even the most advanced analyses rely on stringent long-panel conditions such as $N(\log T)^2/T \to 0$ \citep{galvao2020unbiased}, effectively requiring the time dimension to dominate, which is hardly satisfied in empirical applications. Indeed, \citet{besstremyannaya2019reconsideration} review $81$ empirical papers with panel quantile models and document that only $6$ involve data with $N<T$, whereas $55$ of them have $N\ge 10T$.

In this paper, we develop a novel framework for FEQR models that accommodates common time effects while preserving asymptotic normality as both \(N\) and \(T\) diverge, thereby covering a wide range of asymptotic regimes, including \(N \asymp T\) and \(T \ll N\). The inclusion of common shocks alters both the convergence rate and the form of the asymptotic covariance, rendering existing inference procedures for conventional FEQR inconsistent. This relaxation of conventional rate restrictions is driven by a phenomenon we term \emph{data-generating process (DGP)–induced smoothing}: conditional averaging over latent shocks inherently smooths the empirical criterion and concentrates its stochastic variation around a smooth component, thereby enabling refined stochastic expansions.
Exploiting this structure, we establish uniform convergence and asymptotic normality of the standard unregularised FEQR estimator. In addition, we propose a new simple robust covariance estimator that is consistent under both common-shock and classical FEQR settings. Consequently, valid inference does not require prior knowledge of whether common shocks are present. These results contribute to the broader literature on nonlinear panel models with large panels \((N,T \to \infty)\).

 In cross-sectional settings with common shocks, the logic of \citet{andrews2005cross} implies that quantile regression may be inconsistent. The panel setting, however, fundamentally changes this conclusion. When units are observed repeatedly over time and the common shocks vary independently across periods, temporal averaging attenuates their influence and restores consistency of FEQR. This does not mean that the shocks are asymptotically negligible. Rather, their effect remains first-order, entering the limiting distribution at a rate determined by the time dimension and thereby shaping the estimator's leading asymptotic behaviour.

Another strand of the literature attempts to accommodate common shocks through factor-based approaches to panel quantile regression. \citet{chen2021quantile} develop identification and estimation for quantile factor models with interactive individual and time effects, requiring large \(N\) and \(T\). Related contributions include \citet{AndoBai2020}, who use principal components for heterogeneous parameter models; \citet{MaLintonGao2021}, who study semiparametric quantile factor models; \citet{HardingLamarchePesaran2020}, who analyse dynamic quantile models with factors; \citet{BelloniChenPadillaWang2023}, who consider high-dimensional latent heterogeneity; and \citet{Chen2024}, who proposes a two-step estimator for models with interactive fixed effects. See \cite{fernandez2018fixed} for a thorough review of the factor-based nonlinear panel data.

Our paper differs from this line of work in that we do not impose an explicit factor structure. Instead, within a classical linear quantile regression specification with individual fixed effects first proposed by  \cite{Koenker04}, we allow time-specific common shocks to enter the regressors and the error term in a highly flexible and unspecified manner. This modelling choice preserves the familiar and interpretable fixed-effects quantile regression structure while accommodating rich forms of cross-sectional dependence. The trade-off is that we do not attempt to identify or recover an underlying factor or interactive fixed-effects structure. In our view, this approach strikes a reasonable balance between interpretability and flexibility. At the same time, it remains closely aligned with specifications that are widely utilised by empirical researchers and is simple and straightforward to implement computationally.

The reminder of the paper is organised as follows. Section \ref{sec:model_estimation} presents the baseline model and estimation procedure. In Section \ref{sec:main} we formalise the main statistical properties of the estimator. Section \ref{sec:inference} provides practical inference procedures and establishes their validity. Section \ref{sec:model_estimation_m} presents generalisation of our baseline results to cover serially correlated common shocks. Section \ref{sec:simulation} provides Monte Carlo simulations. Finally, Section \ref{sec:conclusion} concludes.

\subsection*{Notations}

Let $\mathbb{N}$ denote the set of positive integers and $\mathbb{R}$ the real line. 
For a finite-dimensional vector $a$, let $\|a\|$ denote its Euclidean norm, and for a matrix $A$, let $A'$ denote its transpose. The probability spaces
we consider are assumed to be standard Borel.
Throughout the remainder of the paper, all asymptotic statements refer to regimes in which both $N$ and $T$ diverge to infinity.

\medskip

\section{Model and Estimation}\label{sec:model_estimation}
We begin by introducing the modelling framework; further motivation is provided in Remark \ref{rem:DGP} below.
Consider the data-generating process
\begin{align}
    (Y_{it}, X_{it}') = g(A_i, B_t, U_{it}), \label{eq:DGP}
\end{align}
where the regressor vector $X_{it}$ excludes an intercept. The latent variables $\{A_i\}_{i\in\N}$, $\{B_t\}_{t\in\N}$, and $\{U_{it}\}_{(i,t)\in\N^2}$ are mutually independent and unobserved by the econometrician, and $g:\mathcal A\times \mathcal B\times\mathcal U\to \R^{p+1}$ is an unknown Borel-measurable function. For simplicity, in the baseline results we assume that the common time shocks,
$B_t \in \mathcal{B}$, are i.i.d.\ across $t$, and that the idiosyncratic
shocks, $U_{it} \in \mathcal{U}$, are i.i.d.\ across $(i,t)$. Section
\ref{sec:model_estimation_m} extends the analysis to allow for serially
correlated common shocks.  No identical sampling or particular dependence structure is imposed on the unit-specific effects $A_i \in \mathcal{A}$. We impose no smoothness-type restrictions on the function $g$ with respect to these latent shocks. Moreover, the spaces $\mathcal A$, $\mathcal B$, and $\mathcal U$ are allowed to be high-dimensional and structurally complex.

In the spirit of \citet[Equation (2.2)]{Koenker04}, we specify a quantile regression model with individual fixed effects:
\begin{align}
    Q_\tau(Y_{it} \mid X_{it}, A_i) = \alpha_0(\tau; A_i) + X_{it}'\beta_0(\tau). \label{eq:model}
\end{align}
Here, similar to \cite{kato2012asymptotics}, the fixed effects are allowed to vary across quantile levels.
For the purpose of asymptotic analysis, we condition on a realisation of the individual effects, $\{A_i\}_{i=1}^N = \{a_i\}_{i=1}^N$. Under this conditioning, the model can be written as
\begin{align*}
    Y_{it} &= \alpha_{i0}(\tau) + X_{it}'\beta_0(\tau) + \epsilon_{it}, 
    \quad Q_\tau(\epsilon_{it} \mid X_{it}, a_i) = 0,
\end{align*}
where $\alpha_{i0}(\tau) = \alpha_0(\tau; a_i)$. As a consequence, the regressor--error pair then admits a representation,
\begin{align}
        (X_{it}, \epsilon_{it}) = g_i(B_t, U_{it}),\label{eq:DGP_cond}
\end{align}
 where $g_i(b,u) = g(a_i, b, u)$. We impose no parametric restrictions on the relationship between the individual fixed effects $\alpha_{i0}(\tau)$ and the regressors $X_{it}$; in particular, the fixed effects are allowed to vary flexibly with $\tau$. Moreover, under this framework, common shocks are permitted to exert heterogeneous effects across population units, with the magnitude and direction of these effects depending on unit-specific characteristics, which may be either observed or unobserved.

\begin{remark}[The rationale behind the structures of Equations \eqref{eq:DGP} and \eqref{eq:DGP_cond}]\label{rem:DGP}
The nonparametric representation in \eqref{eq:DGP_cond} is closely related to \cite{andrews2005cross}, with similar modelling frameworks considered in \cite{chernozhukov2013average} and \cite{hausman2025linear} in panel data contexts. Fundamentally, it is highly general and can be viewed as an application of the de Finetti’s theorem under exchangeability.
 In particular, for each fixed $i$, de Finetti’s theorem implies that exchangeable observations are conditionally independent given a latent variable, yielding a representation of the form in \eqref{eq:DGP_cond}--see Lemma 2 in Section 7 in \cite{andrews2005cross}. Exchangeability arises naturally in a wide range of economic models. For example, in dynamic oligopoly models with investment, \citet{athey2001investment} show that exchangeable profit functions are consistent with standard environments such as Cournot and differentiated-product competition, reflecting a symmetry restriction whereby firms’ payoffs depend on rivals’ actions and states but not their identities. This symmetry is closely related to the notion of anonymity in cooperative game theory and social choice theory (e.g.\ \citealt{moulin1991axioms}). Similarly, in differentiated product markets, \citet{berry1995automobile} argue that demand and cost functions are exchangeable in competitors’ characteristics and show that equilibrium uniqueness implies exchangeability in both observed and unobserved components.

Likewise, the  factor-type representation in \eqref{eq:DGP} is motivated by the literature on exchangeable arrays and network-type dependence \citep{bickel2011method, graham2024sparse}, as well as by work on two-way clustering and dependence structures in panel data \citep{menzel2021bootstrap,davezies2021empirical, chiang2024standard,athey2025identification}. Under a suitable exchangeability condition, the existence of such data-generating processes is guaranteed by the Aldous--Hoover representation theorem (see Chapter~7 of \citealp{kallenberg2005probabilistic}).
 \hfill $\lozenge$ \end{remark}

Given data $\{(Y_{it}, X_{it}') : i=1,\ldots,N,\; t=1,\ldots,T\}$, we jointly estimate the common parameter  $\beta_0(\tau)$ and the fixed effects $\{\alpha_{i0}(\tau)\}_{i=1}^N$ using the standard Koenker-type FEQR estimator \citep{Koenker04} without regularisation:
\begin{align}
    (\hat{\bm{\alpha}}(\tau), \hat{\beta}(\tau)) 
    = \arg\min_{\{\alpha_i\}_{i=1}^N,\, \beta} 
    \frac{1}{NT} \sum_{i=1}^N \sum_{t=1}^T 
    \rho_\tau\!\left(Y_{it} - \alpha_i - X_{it}'\beta\right),\label{eq:FEQR}
\end{align}
where $\rho_\tau(u) = u\big(\tau - \mathbbm{1}\{u < 0\}\big)$ is the check loss function. Unless stated otherwise, we fix $\tau \in (0,1)$ throughout and suppress its dependence in the notation.

From a computational perspective, the FEQR estimator is straightforward to implement using existing optimisation routines. The objective function is convex in $(\bm{\alpha}, \beta)$, so estimation reduces to a standard linear programming problem. Consequently, the estimator can be readily computed using widely available quantile regression packages that accommodate fixed effects or high-dimensional controls, such as \texttt{quantreg} in \textsf{R}, \texttt{statsmodels} and \texttt{cvxpy} in Python, as well as built-in procedures in standard econometric software. Modern algorithms for large-scale convex optimisation make it feasible to handle panels with a large number of cross-sectional units, ensuring practical scalability in empirically relevant applications.

\section{Asymptotic Theory}\label{sec:main}
In this section, we study the theoretical properties of the FEQR estimator under the common shock framework.
 
We denote by $F_i$ and $f_i$ the distribution function and density of 
$\epsilon_{it}=g_i^\epsilon(B_t,U_{it})$.
Further, let $F_{i,X}(\cdot \mid x)$ and $f_{i,X}(\cdot \mid x)$ denote the conditional 
distribution function and density of $\epsilon_{i1}$ given $X_{i1}=x$, and let 
$F_{i,B}(\cdot \mid b)$ and $f_{i,B}(\cdot \mid b)$ denote the corresponding objects 
given $B_1=b$. Likewise, let $F_{i,XB}(\cdot \mid x,b)$ and $f_{i,XB}(\cdot \mid x,b)$ 
denote the conditional distribution function and density of $\epsilon_{i1}$ given 
$X_{i1}=x$ and $B_1=b$. When no confusion is likely to arise,
we suppress the subscripts $X$ and/or $B$ and write, with slight abuse of notation, 
expressions such as $f_i(\cdot \mid x)$ and $F_i(\cdot \mid x,b)$.
Finally, let
$\mathcal{X}_i$ and $\mathcal{E}_i$ denote the support of $X_{i1}$ and $\epsilon_{i1}$, 
and write $\mathcal{X} = \bigcup_{i \geq 1} \mathcal{X}_i \subset \R^p$ and $\mathcal{E} = \bigcup_{i \geq i} \mathcal{E}_i \subset \R$.

\subsection{Uniform Consistency}\label{sec:consistency}
    We first present the assumptions needed for uniform consistency. 
\begin{assumption}[Bounded regressors] \label{assump:x_bounded_support}
    $\mathcal{X}$ is a bounded set in $\R^p$ for a fixed $p$.
\end{assumption}

\begin{assumption}[Identification] \label{assump:identification}
    For each $\delta>0$, 
    \begin{align*}
        \epsilon_\delta:=\inf_{i\ge 1} \inf_{|\alpha|+\|\beta\|_1}\E\left[\int_{0}^{\alpha+X_{i1}'\beta}\{F_i(s|X_{i1})-\tau\}ds\right]>0,
    \end{align*}
    where $\|\cdot\|_1$ is the $\ell_1$ norm.
\end{assumption}
Assumption \ref{assump:x_bounded_support} is imposed for technical convenience and could in principle be relaxed to sub-gaussian tail bounds or suitable polynomial moment conditions, albeit at the expense of substantially more involved arguments. This assumption is recurrent in the quantile regression literature, see, for example, \citet{Koenker04} and \citet{ChernozhukovHansen06}. Assumption \ref{assump:identification} is a standard identification condition that can be found in, e.g. \cite{fernandez2005bias} and \cite{kato2012asymptotics}. 

Under Assumptions \ref{assump:x_bounded_support} and \ref{assump:identification}, the following lemma shows that both the individual effects estimators $\hat{\alpha}_i$ and the common parameter estimator $\hat{\beta}$ are uniformly consistent for their true values. 
\begin{proposition}[Uniform consistency]\label{prop:consistency}
    Under Assumptions \ref{assump:x_bounded_support} and \ref{assump:identification} and suppose that $(\log N) / T \to 0$, we have
$\max_{1 \leq i \leq N}|\hat{\alpha}_i - \alpha_{i0}| \vee \|\hat{\beta} - \beta_0\| \pto 0$.
\end{proposition}
A proof is provided in Section \ref{sec:proof of prop:consistency} of the Appendix. The argument closely parallels that of Theorem 3.1 in \cite{kato2012asymptotics}, with only minor modifications. In particular, the presence of common shocks does not materially complicate the consistency proof. By contrast, establishing asymptotic normality requires a substantially different line of argument.
\subsection{Asymptotic Distribution}\label{sec:AN}
To establish asymptotic distributional theory, we impose the following additional assumptions.
Write $P_i$ as the joint law of $(\epsilon_i, X_{i1}, B_1)$ on $\mathcal{E}_i \times \mathcal{X}_i \times \mathcal{B}$.
\begin{assumption}[Joint densities] \label{assump:pdf_exists}
    For each $i$, $(\epsilon_{i1}, X_{i1}, B_1)$ admits a joint density
    on $\mathcal{E}_i \times \mathcal{X}_i \times \mathcal{B}$ with respect to
    $\lambda \otimes P_{(X_{i1}, B_1)}$, where $\lambda$ denotes Lebesgue measure on
    $\mathbb{R}$ and $P_{(X_{i1}, B_1)}$ is the marginal distribution of $(X_{i1}, B_1)$.
\end{assumption}

Henceforth, ``for all'' is understood to mean 
$P_i$-almost everywhere.
\begin{assumption}[Smoothness of densities] \label{assump:bounded_pdf_derivative}
 For each $i \in \N$, 
\begin{enumerate}
    \item[(a)]  $f_i(\epsilon \mid x, b)$ is continuously differentiable
    with respect to $\epsilon$ for all $(x, b)$ in $\mathcal{X}_i \times \mathcal{B}$.
    \item[(b)]  There exists finite constant $C_f$ and $L_f$ such that 
    $$f_i(\epsilon \mid x, b) \leq C_f, \quad\text{and } \quad|\partial f_i(\epsilon \mid x,b) / \partial \epsilon| \leq L_f$$
    uniformly over $(\epsilon, x, b) \in \mathcal{E}_i \times \mathcal{X}_i \times \mathcal{B}$ and $i \in \N$;
   \item[(c)] $f_i(0 \mid b)$ is bounded from below by some positive constant independent of $i$ and $b$.
\end{enumerate}
    
\end{assumption}

\begin{remark}
    Assumption
    \ref{assump:pdf_exists} is an if and only if condition for the almost-sure existence of a
    conditional density of $\epsilon_{i1}$ given $(X_{i1}, B_1)$.
    Note that the distribution of $X_{i1}$ may have point masses, so the
    regressors are allowed to be discrete-valued. Moreover, we do not require
    $\mathcal{B}$ to be finite-dimensional. Since $\mathcal{E}_i \subset \mathbb{R}$ is a standard
    Borel space, a regular conditional distribution of $\epsilon_{i1}$ given
    $(X_{i1}, B_1)$ always exists; see, for example, Theorem 8.5 of
\cite{kallenberg2021foundation}. Assumption \ref{assump:pdf_exists} ensures that this conditional distribution is
    absolutely continuous with respect to Lebesgue measure for $P_{(X_{i1}, B_1)}$-almost every $(x, b)$.
\end{remark}

Assumptions \ref{assump:pdf_exists} and \ref{assump:bounded_pdf_derivative} are mild, and primarily ensure the existence and sufficient smoothness of the conditional density of $\epsilon_{i1}$ given $(X_{i1}, B_1)$. 

Note that the existence of the conditional density 
of $\epsilon_{i1}$ given $(X_{i1}, B_1)$ implies the 
existence of the conditional density 
of $\epsilon_{i1}$ given $X_{i1}$ and the 
marginal density of $\epsilon_{i1}$.
Let us define $\gamma_i = \E[f_i(0 \mid X_{i1})X_{i1}] / f_i(0)$, which plays a key role in the expression of the asymptotic covariance  introduced below.

\begin{assumption}[Variance with common shocks] \label{assump:variance_bounded_gamma}
 Suppose that  $\sup_{i\in \mathbb N} \|\gamma_i\|<\infty$.  Further assume the following matrix exists and is nonzero
    \begin{align}
     \Sigma=   \lim_{N \to \infty} \var\left( \frac{1}{N} \sum_{i = 1}^N \E[(\tau - \idf\{\epsilon_{i1} \leq 0\})(X_{i1} - \gamma_i)\mid B_1]\right) .\label{eq:Sigma}
    \end{align}
\end{assumption}

Assumption \ref{assump:variance_bounded_gamma} implicitly requires that the common time effect has a non-negligible impact on the distribution of the data. This condition rules out the classical fixed-effects quantile regression framework analysed in \citet{kato2012asymptotics} and subsequent studies, where—after controlling for individual fixed effects—the disturbances are independent across $(i,t)$ pairs, corresponding to a degenerate special case of our setting. The assumption is conceptually related to the non-degeneracy conditions imposed in the multiway clustering literature; see, for example, \citet{chiang2024standard}.

In the cross-sectional setting without panel data, \citet{andrews2005cross} shows that the OLS estimator is inconsistent when the common shock generates a nonzero conditional covariance between the regressor and the error term. This is because the consistency of OLS relies on 
$\E[X_i \epsilon_i] = 0$, and hence if $X_i \epsilon_i$ systematically depends on the common shock, OLS is no longer consistent. In our notation, this can be written as
$\var\left( \E[X_i \epsilon_i \given B]\right) \neq 0$, where $B$ is the common shock.
The anagolous object in quantile regression is the 
score $(\tau - \idf\{\epsilon_i \leq 0\})X_i$, so the 
corresponding condition would be 
$\var\left( \E[(\tau - \idf\{\epsilon_i \leq 0\}) X_i\given B]\right) \neq 0$.

In our setting, we allow the joint distribution 
of $(X_{it}, \epsilon_{it})$ to vary across $i$, 
and so the analogous assumption in panel data quantile regression is captured by \eqref{eq:Sigma}.
Moreover, in panel data, the availability of repeated observations over time preserves the consistency of 
$\hat{\beta}$, as shown in Proposition \ref{prop:consistency}. 
However, while consistency is retained, common shocks continue to affect the estimator through its asymptotic distribution.

\begin{remark}[Structure of asymptotic covariance]\label{rem:covariance_structure}
This term $\Sigma$ constitutes the middle component of the sandwich-form asymptotic covariance of $\hat\beta$. In contrast to the corresponding middle component in \cite{kato2012asymptotics} and \cite{galvao2020unbiased}, which, in our notation, is 
\begin{align}
\Omega=\lim_{N \to \infty}\frac{1}{N}\sum_{i=1}^N 
\var \bigl((\tau - \idf\{\epsilon_{i1}\le 0\})(X_{i1}-\gamma_i)\bigr),\label{eq:meat_classical}    
\end{align}
it admits a fundamentally different structural representation.

To see this, consider the variance term appearing inside the limit in the definition of $\Sigma$:
\begin{align*}
&\var\!\left( \frac{1}{N} \sum_{i = 1}^N 
\E[(\tau - \idf\{\epsilon_{i1} \le 0\})(X_{i1} - \gamma_i)\mid B_1]\right)\\
&= \frac{1}{N^2} \sum_{i = 1}^N 
\var\!\left( \E[(\tau - \idf\{\epsilon_{i1} \le 0\})(X_{i1} - \gamma_i)\mid B_1]\right) \\
&\quad + \frac{1}{N^2} \sum_{i = 1}^N\sum_{j\ne i} 
\cov\!\Big( \E[(\tau - \idf\{\epsilon_{i1} \le 0\})(X_{i1} - \gamma_i)\mid B_1], 
 \E[(\tau - \idf\{\epsilon_{j1} \le 0\})(X_{j1} - \gamma_j)\mid B_1] \Big).
\end{align*}
The distinction is not merely the presence of the cross-sectional covariance terms. Even the first sum differs in structure. By the law of total variance,
\begin{align*}
\var\!\left( \E[(\tau - \idf\{\epsilon_{i1} \le 0\})(X_{i1} - \gamma_i)\mid B_1]\right)
&= \var\!\left((\tau - \idf\{\epsilon_{i1} \le 0\})(X_{i1} - \gamma_i)\right) \\
&\quad - \E\!\left[\var\!\left((\tau - \idf\{\epsilon_{i1} \le 0\})(X_{i1} - \gamma_i)\mid B_1\right)\right].
\end{align*}
Unless the conditional variance in the second term on the right-hand side vanishes, the aggregate variance contribution differs from the i.i.d.-type structure in \cite{kato2012asymptotics}. The common time shocks therefore alter the structure of the covariance term even before accounting for cross-unit covariance.
 \hfill $\lozenge$ \end{remark}


\begin{assumption}[Jacobian matrix]\label{assump:Jacobian}
       The following matrix exists and is invertible
    \begin{align*}
     \Gamma=  \lim_{N \to \infty} \frac{1}{N}\sum_{i = 1}^N \E[f_i(0 \mid X_{i1})X_{i1}(X_{i1} - \gamma_i)'] .
    \end{align*}

\end{assumption}
Assumption \ref{assump:Jacobian} requires existence of the Jacobian matrix and the usual full-rank condition on the bread matrix in the sandwich-form asymptotic covariance. In contrast to the middle component $\Sigma$ in Assumption \ref{assump:variance_bounded_gamma}, $\Gamma$ retains the same structure as its counterpart in classicial FEQR theory as in \cite{kato2012asymptotics}.
    Also that $\Gamma$ is a symmetric matrix--observe that $\E[f_i(0 \mid X_{i1})\gamma_{i}(X_{i1} - \gamma_i)'] = 0$ for all $i$, which implies $\E[f_i(0 \mid X_{i1})X_{i1}(X_{i1} - \gamma_i)'] =  \E[f_i(0 \mid X_{i1})(X_{i1} - \gamma_i)(X_{i1} - \gamma_i)'].$

Denote the subgradients of the check function evaluated at $(\bm{\alpha}', \beta')$ by
\begin{align*}
    \mathbb{H}_{Ni}^{(1)}(\alpha_i, \beta) &=  \frac{1}{T} \sum_{t = 1}^T (\tau - \idf\{Y_{it} \leq \alpha_i + X_{it}'\beta\}), \\
    \mathbb{H}_{N}^{(2)}(\bm{\alpha}, \beta) &= \frac{1}{NT} \sum_{i = 1}^N \sum_{t = 1}^T (\tau - \idf\{Y_{it} \leq \alpha_i + X_{it}'\beta\}) X_{it}.
\end{align*}

In classical FEQR environments, such as in \cite{kato2012asymptotics}, the non-differentiability of the check loss precludes standard Taylor expansions. Combined with the faster $\sqrt{NT}$-rate implied by cross-sectional independence, the empirical process terms arising from differences between subgradients evaluated at the estimator and at the true parameter become difficult to control, thereby typically requiring asymptotic regimes in which $T$ grows substantially faster than $N$.
 Under the common time-effect structure in \eqref{eq:DGP_cond}, however, these non-smooth subgradients admit approximations via suitably defined differentiable projections. This smoothness  restores local differentiability, thereby permitting Taylor expansions and substantially simplifying the arguments and weaken the conditions.

Explicitly, let $\bm{B}_T = \{B_t\}_{t=1}^T$ and define the conditional (on common shocks) projections of the subgradients as
\begin{align*}
    \tilde{\mathbb{H}}_{Ni}^{(1)}(\alpha_i, \beta)
    &= \E\!\left[\mathbb{H}_{Ni}^{(1)}(\alpha_i, \beta) \mid \bm{B}_T\right]
      = \frac{1}{T} \sum_{t = 1}^T \E\!\left[\tau - \idf\{Y_{it} \le \alpha_i + X_{it}'\beta\} \mid B_t\right], \\
    \tilde{\mathbb{H}}_{N}^{(2)}(\bm{\alpha}, \beta)
    &= \E\!\left[\frac{1}{NT} \sum_{i=1}^N \sum_{t=1}^T
       (\tau - \idf\{Y_{it} \le \alpha_i + X_{it}'\beta\}) X_{it}
       \;\middle|\; \bm{B}_T \right] \\
    &= \frac{1}{NT} \sum_{i=1}^N \sum_{t=1}^T
       \E\!\left[(\tau - \idf\{Y_{it} \le \alpha_i + X_{it}'\beta\}) X_{it} \mid B_t\right].
\end{align*}

Conditioning on $\bm{B}_T$ removes the non-smoothness stemming from the indicator function at the level of the stochastic fluctuations. The projected objects $\tilde{\mathbb{H}}_{Ni}^{(1)}$ and $\tilde{\mathbb{H}}_{N}^{(2)}$ are expectations of indicator functions and therefore can be expressed in terms of conditional distribution functions of $Y_{it}$ given $B_t$. Under mild regularity conditions—such as conditional smoothness of these distributions and bounded densities—the projections approximate the original subgradients uniformly as $N$ diverges while, at the same time, being smooth functions of $(\alpha_i,\beta)$. This structure permits standard Taylor-type expansions around the true parameters despite the non-smoothness of the original objective function, while creating a limiting distribution driven by the common shocks, it yields simple and general asymptotic theory. We refer to this phenomenon as \emph{DGP-induced smoothing}: conditional averaging over the common time factors smooths the empirical criterion through the randomness of the idiosyncratic components, resulting in stochastic fluctuations of order $\sqrt{T}$ around its mean.
 This mechanism is central to obtaining refined stochastic expansions and, ultimately, to establishing asymptotic normality under panel asymptotic regimes that allow $N$ to be large relative to $T$.


We now present the main asymptotic distributional result.
\begin{theorem}[Asymptotic distribution]\label{thm:main}
Suppose Assumptions \ref{assump:x_bounded_support}--\ref{assump:Jacobian} hold and suppose that $(\log N)^2 / T \to 0$, then we have
\begin{align*}
    \sqrt{T}(\hat{\beta} - \beta_0) \longdto \mathcal{N}(0, V),\quad\text{where } V= \Gamma^{-1}\Sigma \Gamma^{-1}.
\end{align*}
\end{theorem}
A proof can be found in Section \ref{sec:proof of thm:main} in the appendix. Some remarks are in order.

\begin{remark}[Comparison with classical FEQR results]
   Compared with the asymptotic normality results for FEQR obtained in the absence of time effects, such as \cite{kato2012asymptotics} and \cite{galvao2020unbiased}, several important differences arise. First, the presence of pervasive common shocks changes the effective convergence rate from $\sqrt{NT}$ to $\sqrt{T}$. Despite of the slower rate, this is not necessarily detrimental, as the bias induced by the incidental parameter problem is attenuated in our setting and has a less pronounced impact on inference. Furthermore, standard inference procedure for FEQR,  therefore do not deliver a consistent estimate of the asymptotic covariance.  Second, because the projected objective is smooth, Taylor expansion arguments become available. As a consequence, the resulting linear representation—and hence the asymptotic covariance—explicitly incorporates the effect of estimating the fixed effects. Finally, under the classical setup, even the refined analysis in \citet{galvao2020unbiased} imposes the restriction \( N(\log T)^2 / T \to 0 \). By comparison, our requirement \( (\log N)^2 / T \to 0 \) is markedly weaker, thereby allowing for a wide class of asymptotic sequences with \( T \ll N \).
 \hfill $\lozenge$ \end{remark}

\begin{remark}[Relationship with \cite{andrews2005cross}]
    \citet{andrews2005cross} shows that, in a cross-sectional setting, the OLS estimator is inconsistent when a common shock induces a nonzero conditional covariance between the regressor and the error term. Specifically, under a common shock $B$,
    \begin{align*}
        \hat{\beta} \longpto \beta_0 + \gamma(B),
    \end{align*}
    $\gamma(B) \neq 0$ whenever $\var\left( \E[X_i \epsilon_i \given B] \right) \neq 0$.
    In contrast, Proposition \ref{prop:consistency} shows that in the panel setting, aggregating observations over time with i.i.d. $\{B_t\}$ averages out the effect of common shocks, so that consistency of FEQR estimator is retained.
    Theorem \ref{thm:main} further shows that this effect remains first-order asymptotically relevant: it appears with $\sqrt{T}$ rate, and dominates the first order asymptotics of the FEQR estimator.
\end{remark}

\begin{remark}[Intuition behind Theorem \ref{thm:main}]
  Here we provide intuition for the proof. The key step is to work with the
\emph{projected score}, obtained by conditioning on the common time effects:
\begin{align*}
 \tilde{\mathbb{H}}_{Ni}^{(1)}(\alpha_i, \beta)
    & 
    = \frac{1}{T} \sum_{t = 1}^T 
       \E\!\left[\tau - \mathbbm{1}\{Y_{it} \le \alpha_i + X_{it}'\beta\} \mid B_t\right].
\end{align*}
We show that this projected score is uniformly close to the original score
$\mathbb{H}^{(1)}_{Ni}(\alpha_i,\beta)$. Furthermore, under mild conditions, $\tilde{\mathbb{H}}_{Ni}^{(1)}(\alpha_i,\beta)$ is
a differentiable function of $(\alpha_i,\beta)$, even though the original score is
non-differentiable.
This DGP-induced smoothing  permits a Taylor expansion around the true parameter, in contrast to Equation~(A.6) of \citet{kato2012asymptotics}, where the absence of smoothness precludes such a linearisation:
\begin{align*}
 \hat{\alpha}_i - \alpha_{i0}
=& \left( \frac{1}{T} \sum_{t = 1}^T f_i(0 \mid B_t) \right)^{-1}
   \tilde{\mathbb{H}}_{Ni}^{(1)}(\alpha_{i0}, \beta_0) \\
&- \left( \frac{1}{T} \sum_{t = 1}^T f_i(0 \mid B_t) \right)^{-1}
   \left( \frac{1}{T} \sum_{t = 1}^T 
   \E[f_i(0 \mid X_{it}, B_t)X_{it}' \mid B_t] \right)
   (\hat{\beta} - \beta_0)  \\
&+ \text{higher-order terms}.
\end{align*}
Such an expansion is not available in earlier FEQR analyses,
where the lack of smoothness of the objective necessitated a coarser rate
based on concentration inequalities. The projected-score approach, therefore
replaces non-smooth empirical process arguments with a smooth stochastic
expansion. Together with the $\sqrt{T}$-asymptotic normality induced by the common shocks, this permits substantially weaker restrictions on the relative growth rates of $N$ and $T$
 \hfill $\lozenge$ \end{remark}

\section{Statistical Inference}\label{sec:inference}

This section suggests simple practical methods for inference for the quantile model with common shocks. To utilise Theorem~\ref{thm:main} for statistical inference, one needs a consistent estimator of the covariance matrix $V$. We now introduce the new covariance estimator and establish its asymptotic guarantees.

For a kernel function (a probability density function) $K:\R\to \R$, denote  $K_h(u)=h^{-1}K(u/h)$ and $\hat \epsilon_{it}=Y_{it}-\hat\alpha_i-X_{it}'\hat\beta$. We define
\begin{align}
    \hat m_{Nt}=\frac{1}{N}\sum_{i=1}^N (\tau- \idf\{\hat \epsilon_{it}\le 0\})(X_{it}-\hat 
\gamma_i),\quad \bar m_{NT}=\frac{1}{T}\sum_{t=1}^T \hat m_{Nt}. \label{eq:mNt}
\end{align}
where
\begin{align*}
    \hat \gamma_i=\frac{1}{\hat f_i T}\sum_{t=1}^T K_{h}(\hat \epsilon_{it})X_{it},\quad  \hat f_i=\frac{1}{ T}\sum_{t=1}^T K_{h}(\hat \epsilon_{it}).
\end{align*}
We estimate $\Sigma$ and $\Gamma$ by 
 \begin{align}
    \hat \Sigma_N=& \frac{1}{T}\sum_{t=1}^T (\hat m_{Nt} -\bar m_{NT} )(\hat{m}_{Nt} - \bar{m}_{NT})',\quad
    \hat \Gamma_N=\frac{1}{NT}\sum_{i=1}^N\sum_{t=1}^T K_h(\hat \epsilon_{it})X_{it}(X_{it}-\hat \gamma_i),
    \label{eq:Sigma_hat_and_Gamma_hat}
\end{align}
Hence, the asymptotic covariance in Theorem \ref{thm:main} can be estimated by the robust covariance estimator:
\begin{align}
    \hat V= \hat \Gamma_N^{-1} \hat \Sigma_N   \hat \Gamma_N^{-1}.\label{eq:robust_covariance}
\end{align}

We emphasise that this estimator in \eqref{eq:robust_covariance} differs from the standard covariance estimators commonly used in the literature; see, for example, \citet{kato2012asymptotics}. As discussed in Remark~\ref{rem:covariance_structure}, the presence of common shocks fundamentally alters the structure of the asymptotic covariance relative to the classical setting with cross-sectional independence. Consequently, conventional covariance estimators are generally inconsistent in this environment and lead to invalid inference. We illustrate this discrepancy in the simulation study reported in Section \ref{sec:simulation}. The following result establishes the asymptotic properties of the proposed covariance estimator.

\begin{theorem}[Robust covariance estimation]\label{thm:covariance_estimation}
Suppose the kernel $K$ is continuous, bounded, and of bounded variation on the real line, and the sequence of $h$ satisfies $h\to 0$ and $(\log N)/(Th)\to 0$ as $N,T\to \infty$, then  
    \begin{enumerate}
        \item[(i)] Under Assumptions of Theorem \ref{thm:main} and $\log T/N\to 0$, it holds that
        \(
            \hat V \xrightarrow{p} V .
        \)
        \item[(ii)] In the absence of common shocks $\{B_t\}_t$, under the assumptions of Theorem~1 in \citet{galvao2020unbiased},  it holds that
        \(
            N\hat V \xrightarrow{p} \tau(1-\tau)\Gamma^{-1}\Omega\Gamma^{-1},
        \)
        where $\Omega$ is as defined in \eqref{eq:meat_classical}.
    \end{enumerate}
\end{theorem}
A proof can be found in Section \ref{sec:proof of thm:covariance_estimation}. A couple of remarks are in order.
 
\begin{remark}[Estimating $\Sigma$]
 Recall that $\Sigma$ defined in \eqref{eq:Sigma} contains the conditional expectation
\(
\E\!\left[
(\tau - \idf\{\epsilon_{i1} \leq 0\})(X_{i1} - \gamma_i)
\mid B_1
\right]
\).
To estimate it, in addition to the need to replace $\epsilon_{it}$ and $\gamma_i$ by their estimators, we also need to approximate the conditional expectation given the common shock $B_1$. Direct estimation of this object is challenging because $B_t$ is unobserved and may be high-dimensional, rendering standard nonparametric estimation challenging. Fortunately, for each fixed $t$, all cross-sectional units share the same realisation of $B_t$. This allows the conditional expectation given $B_t$ to be consistently approximated by its cross-sectional average over $i=1,\ldots,N$, thereby avoiding the need to estimate it nonparametrically.
 \hfill $\lozenge$ \end{remark}

\begin{remark}[Robustness of the covariance estimator]
Theorem~\ref{thm:covariance_estimation} implies that, even in the absence of common shocks, inference based on the proposed variance estimator $\hat V$ remains valid under the classical fixed effects quantile regression (FEQR) framework of \cite{galvao2020unbiased}. In particular, valid inference does not require the practitioner to determine ex ante whether common shocks are present in the data-generating process. For this reason, we recommend the use of our robust covariance estimator of \eqref{eq:robust_covariance} in empirical applications in place of existing alternatives.
 \hfill $\lozenge$ \end{remark}

\medskip

\section{Extension: Serially Correlated Common Shocks} \label{sec:model_estimation_m}
In this section, we extend our results to the case where $\{B_t\}$ is stationary and
$M$-dependent: that is, $B_t$ and $B_{t+k}$ are independent whenever $k > M$, for some fixed positive integer $M$, while
arbitrary dependence is allowed for lags $k \le M$. All other modelling assumptions are identical to those in Section
\ref{sec:model_estimation}.

\subsection{Uniform Consistency}
Under serially correlated common shocks, the consistency result 
holds under the same conditions as before.

\begin{proposition} \label{prop:consistency_m}
    Suppose $\{ B_t \}$ is stationary and $B_t$ and $B_{t + k}$ are independent for any $t $ and $k >M$.
    Under Assumptions \ref{assump:x_bounded_support} and \ref{assump:identification} and suppose that $(\log N) / T \to 0$, we have $\max_{1 \leq i \leq N}|\hat{\alpha}_i - \alpha_{i0}| \vee \|\hat{\beta} - \beta_0\| \pto 0$.
\end{proposition}
A proof can be found in Section \ref{sec:consistency}.

\subsection{Asymptotic Distribution}
\begin{assumption} \label{assump:variance_bounded_gamma_m}
    Suppose that  $\sup_{i\in \mathbb N} \|\gamma_i\|<\infty$.  Further assume the following matrix exists and is nonzero
    \begin{align}
     \Lambda =   \lim_{N \to \infty}\left( \Sigma^0_N + \sum_{l = 1}^M \left(\Sigma_N^l + {\Sigma_N^l}' \right) \right), \label{eq:Sigma_m}
    \end{align}
    where 
    \begin{align*}
        J_{Nt} &= \frac{1}{N} \sum_{i = 1}^N \E[(\tau - \idf\{\epsilon_{it} \leq 0\})(X_{it} - \gamma_i) \given B_t], \\
        \Sigma_N^l &= \cov \left( J_{N1}, J_{N(1 + l)}\right) = \E\left[J_{N1} J_{N(1 + l)}'\right].
    \end{align*}
\end{assumption}

\begin{theorem}[Asymptotic distribution]\label{thm:main_m}
Suppose $\{ B_t \}$ is stationary and $B_t$ and $B_{t + k}$ are independent for any $t$ and $k > M$.
Suppose Assumptions \ref{assump:x_bounded_support}--\ref{assump:bounded_pdf_derivative} and \ref{assump:Jacobian}--\ref{assump:variance_bounded_gamma_m} hold and suppose that $(\log N)^2 / T \to 0$, then we have
\begin{align*}
    \sqrt{T}(\hat{\beta} - \beta_0) \longdto \mathcal{N}(0, V),\quad\text{where } V= \Gamma^{-1}\Lambda \Gamma^{-1}.
\end{align*}
\end{theorem}
A proof can be found in section \ref{sec:proof of thm:main}.

\subsection{Statistical Inference}
For a kernel function $K: \R \to \R$, denote $K_h(u) = h^{-1} K(u/h)$ and 
$\hat{\epsilon}_{it} = Y_{it} - \hat{\alpha}_i - X_{it}'\hat{\beta}$.
Define $\hat{\gamma}_i$, $\hat{f}_i$, $\hat{m}_{Nt}$ and $\overline{m}_{NT}$ as in Equation \eqref{eq:mNt}.
Define for all $1 \leq l \leq M$, 
\begin{align*}
    \hat{\Sigma}_N^l \coloneq \frac{1}{T - l} \sum_{t = 1}^{T - l} \left( \hat{m}_{Nt} - \bar{m}_{NT}\right) \left( \hat{m}_{N(t + l)} - \bar{m}_{NT}\right)'.
\end{align*}
We estimate $\Gamma$ and $\Lambda$ by 
\begin{align*}
    \hat{\Gamma}_N \coloneq \frac{1}{NT} \sum_{i = 1}^N \sum_{t = 1}^T K_h(\hat{\epsilon}_{it})X_{it}(X_{it} - \hat{\gamma}_i), \quad 
    \hat{\Lambda}_N \coloneq \hat{\Sigma}_N^0 + \sum_{l = 1}^M \left( \hat{\Sigma}_N^l + (\hat{\Sigma}_N^l)' \right).
\end{align*}
It follows that the asymptotic variance in Theorem \ref{thm:main_m} can be estimated by the robust covariance estimator:
\begin{align}
    \hat V= \hat \Gamma_N^{-1} \hat \Lambda_N   \hat \Gamma_N^{-1}.\label{eq:robust_covariance_m}
\end{align}
\begin{theorem} \label{thm:covariance_estimation_m}
    Suppose $\{ B_t \}$ is stationary and $B_t$ and $B_{t + k}$ are independent for any $t $ and $k > M$.
    Suppose the kernel $K$ is continuous, bounded, and of bounded variation on the real line, and the sequence of $h$ satisfies $h\to 0$ and $(\log N)/(Th)\to 0$ as $N,T\to \infty$.  Under Assumptions of Theorem \ref{thm:main_m} and $\log T/N\to 0$, it holds that $\hat{V} \longpto V$.
\end{theorem}
A proof can be found in section \ref{sec:proof of thm:covariance_estimation}.

\medskip

\section{Monte Carlo Simulations}
\label{sec:simulation}

We conduct Monte Carlo simulations to examine the finite-sample performance of the fixed effects quantile regression estimator and to evaluate the accuracy of the proposed variance estimator in the presence of common time shocks. The simulation design reflects empirically relevant panel dimensions with a large cross-sectional dimension and a moderate time dimension.

We consider the following location-scale shift model similar to \cite{kato2012asymptotics}
\begin{align}\label{eq:MC DGP}
Y_{it}
=
\alpha_i
+
\beta X_{it}
+
(1 + \gamma X_{it}) U_{it},
\qquad
i = 1,\ldots,N,\quad t = 1,\ldots,T,
\end{align}
where $\{\alpha_i\}_{i=1}^N$ are individual fixed effects and $X_{it}$ is a scalar regressor.
The regressors are generated according to
$
X_{it}
=
\chi^2_{it}(3)
+
0.3\,\alpha_i,
$
where $\chi^2_{it}(3)$ are independent chi-square random variables with three degrees of freedom, and $\alpha_i \sim \mathrm{Uniform}(0,1)$.
The error term is generated as
\begin{align*}
U_{it}
=\frac{\varepsilon_{it} + \eta_t}{\sqrt{2}}
\end{align*}
where the idiosyncratic component $\varepsilon_{it}  \sim N(0,1)$ and common shock $\eta_t \sim N(0,1)$ are mutually independent and i.i.d. across $i$ and $t$. The common shock component $\eta_t$ induces cross-sectional dependence while preserving the conditional quantile structure.
We focus on parameters
$
\beta = 1,
$
$
\gamma = 0.2
$, and estimate the model for three quantile levels $\tau = \{0.25, 0.50, 0.75\}$.
Under this specification, the true quantile slope coefficient equals
$\beta(\tau)
=
\beta + \gamma q_\tau,$
where $q_\tau$ denotes the $\tau$-quantile of the standard normal distribution.

The panel dimensions are set to
$
N \in\{ 250, 500, 1,000\}$,
$
T \in\{ 25,50\}
$,
and each experiment is repeated over $2,000$ Monte Carlo replications.
For each simulated sample, we estimate the FEQR model using the standard Koenker-type estimator as defined in \eqref{eq:FEQR}.
We consider two covariance estimators:
(i) The conventional sandwich variance estimator from \cite{kato2012asymptotics}.
(ii)
The proposed robust covariance estimator of \eqref{eq:robust_covariance}. In both covariance estimators, a Gaussian kernel is used with the bandwidth chosen according to Silverman's rule-of-thumb, $h=1.06 \cdot \mathrm{sd}(\hat{\epsilon})\cdot T^{-1/5}$, while imposing a lower bound of $0.05$ to prevent undersmoothing and ensure numerical stability in finite samples.
For both variance estimators, we construct nominal $95\%$ confidence intervals using asymptotic normal approximations.
We evaluate FEQR estimator performance using bias, root mean squared error (RMSE), and the two covariance estimators by their coverage probabilities. The results are displayed in Tables \ref{tab:mc_bias_rmse_split} and \ref{tab:mc_coverage_split}.

\begin{table}[t]
\centering
\begin{threeparttable}
\caption{Monte Carlo performance of the FEQR  estimator: bias and RMSE}
\label{tab:mc_bias_rmse_split}
\begin{tabular}{lcccccc}
\toprule
&\multicolumn{2}{c}{$\tau=0.25$}&\multicolumn{2}{c}{$\tau=0.50$}&\multicolumn{2}{c}{$\tau=0.75$}\\
\cmidrule(lr){2-7}
$(N,T)$ & Bias & RMSE  & Bias & RMSE   & Bias  & RMSE   \\
\midrule
\addlinespace[0.25em]
$(100,10)$ & 0.0103 & 0.0617 & -0.0005 & 0.0594 & -0.0127 & 0.0641\\
$(100,25)$ & 0.0048 & 0.0388 & 0.0009 & 0.0366 & -0.0043 & 0.03884\\
$(100,50)$ & 0.0014 & 0.0281 & -0.0004 & 0.0265 & -0.0021 & 0.0278\\
$(100,100)$ & 0.0006 & 0.0196 & -0.0005 & 0.0184 & -0.0011 & 0.0194\\
$(250,10)$  & 0.0123 & 0.0558 &  0.0002 &  0.0521 & -0.0102 & 0.0554 \\
$(250,25)$  & 0.0034 & 0.0342 &  -0.0004 &  0.0323 & -0.0043 & 0.0344 \\
$(250,50)$  & 0.0027 &  0.0248 & 0.0007 & 0.0232 &  -0.0023 &  0.0244 \\
$(250,100)$ & 0.0011 & 0.0172 & 0.0003 & 0.0163 & -0.0007 & 0.0172 \\
$(500,10)$  & 0.0099 & 0.0523 & -0.0001 & 0.0498 &  -0.0111 &  0.0527 \\
$(500,25)$  & 0.0035 & 0.0325 & -0.0009 & 0.0311 &  -0.0033 &  0.0325 \\
$(500,50)$  & 0.0004 & 0.0235 & -0.0006 & 0.0230 &  -0.0053 & 0.0271\\
$(500,100)$ & 0.0008 & 0.0161 & 0.0001 & 0.0154 & -0.0010 & 0.0162\\
$(1000,10)$ & 0.0112 & 0.0501 & 0.0003 & 0.0472 & -0.0105 & 0.0511 \\
$(1000,25)$ & 0.0027 & 0.0319 & -0.0009 & 0.0315 & -0.0048 & 0.0334\\
$(1000,50)$ & 0.0019 & 0.0227 & 0.0005 & 0.0216 & -0.0015 & 0.0227\\
$(1000,100)$ & 0.0006 & 0.0153 & 0.0001 & 0.0145 & -0.0008 & 0.0152\\
\bottomrule
\end{tabular}
\begin{tablenotes}[flushleft]\footnotesize
\item \emph{Notes:} Entries report the Monte Carlo bias $\E[\hat\beta(\tau)]-\beta(\tau)$ and RMSE $\sqrt{\E[(\hat\beta(\tau)-\beta(\tau))^2]}$ of the Koenker-type FEQR  estimator at quantile index $\tau$. The true coefficient is $\beta(\tau)=\beta+\gamma q_\tau$. Results are based on $2,000$ Monte Carlo replications under the DGP in Section~\ref{sec:simulation}. 
\end{tablenotes}
\end{threeparttable}
\end{table}

The performance of the FEQR estimator under the common shock model is evaluated first by its bias and RMSE.  Table \ref{tab:mc_bias_rmse_split}  reports the results for the three quantiles and different sample sizes. The results show numerical evidence that the FEQR estimator has small bias for every sample under consideration. Nevertheless, we highlight the important numerical evidence that for a fixed cross-section dimension, the bias decreases as the time-series dimension increases. On the other hand, for a given time-series, the bias is relatively stable as the cross-section dimension increases. Here we recall that from Lemma \ref{prop:consistency} above, the main condition on the sample size growth for establishing consistency of the estimator is that $(\log N)^2 / T \to 0$. Hence, the numerical results corroborate the theoretical prediction. In addition, the RMSE decreases monotonically as either $N$ or $T$ increases. These results suggest that the FEQR estimator performs well in small samples for the common shock model, and even for the case of large $N$ relative to $T$ the model.

\begin{table}[t]
\centering
\begin{threeparttable}
\caption{Monte Carlo inference accuracy: coverage of robust and standard covariance estimators}
\label{tab:mc_coverage_split}
\begin{tabular}{lcccccc}
\toprule
 &\multicolumn{2}{c}{$\tau=0.25$}&\multicolumn{2}{c}{$\tau=0.50$}&\multicolumn{2}{c}{$\tau=0.75$}\\
\cmidrule(lr){2-7}
 $(N,T)$ & Robust & Standard & Robust & Standard & Robust & Standard \\
\midrule
\addlinespace[0.25em]
$(100,10)$ & 0.886 & 0.833 & 0.901 & 0.828 & 0.869 & 0.805\\
$(100,25)$ & 0.918 & 0.806 & 0.924 & 0.827 & 0.914 & 0.806\\
$(100,50)$ & 0.931 & 0.793 & 0.928 & 0.795 & 0.928 & 0.800\\
$(100,100)$ & 0.940 & 0.789 & 0.942 & 0.801 & 0.939 & 0.808\\
$(250,10)$  & 0.865 & 0.648 & 0.898 & 0.688 & 0.868  & 0.671  \\
$(250,25)$  & 0.901 & 0.670 & 0.914 & 0.670 & 0.908  & 0.644  \\
$(250,50)$  & 0.920 & 0.647 & 0.923 & 0.647 & 0.925 &  0.651 \\
$(250,100)$  & 0.932 & 0.630 & 0.935 & 0.636 & 0.933 &  0.647 \\
$(500,10)$  & 0.869 & 0.528 & 0.905 & 0.537 & 0.861 & 0.523 \\
$(500,25)$  & 0.895 & 0.527 & 0.915 & 0.521 & 0.903 & 0.503 \\
$(500,50)$  & 0.922 & 0.527 & 0.925 & 0.509 & 0.920 & 0.508 \\
$(500,100)$  & 0.929 & 0.501 & 0.935 & 0.508 & 0.929 &  0.511 \\
$(1000,10)$ & 0.874 & 0.413 & 0.886 & 0.436 & 0.854 & 0.428 \\
$(1000,25)$ & 0.902 & 0.393 & 0.907 & 0.387 & 0.894 & 0.374 \\
$(1000,50)$ & 0.923 & 0.381 & 0.919 & 0.377 & 0.932 & 0.395 \\
$(1000,100)$  & 0.936 & 0.402 & 0.936 & 0.407 & 0.934 & 0.393 \\
\bottomrule
\end{tabular}
\begin{tablenotes}[flushleft]\footnotesize
\item \emph{Notes:} Entries report empirical coverage probabilities of nominal 95\% confidence intervals for $\beta(\tau)$ constructed using (i) the proposed robust covariance estimator and (ii) the conventional sandwich covariance estimator that ignores cross-sectional dependence.  Results are based on $2,000$ Monte Carlo replications under the DGP in Section~\ref{sec:simulation}. Bandwidth: $h=1.06 \cdot \mathrm{sd}(\hat{\epsilon})\cdot T^{-1/5}$.
\end{tablenotes}
\end{threeparttable}
\end{table}

Results evaluating the finite sample performance of the proposed inference procedure are collected in Table \ref{tab:mc_coverage_split}. If the asymptotic inference procedure correctly approximates the finite sample distribution of $\hat\beta(\tau) - \beta_{0}(\tau)$, the coverage rate should be close to the nominal level of significance ($95$\%). Results show that, for a given $N$, empirical coverage improves as $T$ increases. Moreover, as one investigates the diagonal of the table, for example, the pairs $(100,10)$, $(250,25)$, $(500,50)$, and $(1000,100)$, one sees that coverage rates for the Robust case improve, while empirical coverages worsen for the Standard case. 
Overall, the confidence intervals display accurate finite-sample coverage, with empirical coverage under the proposed variance estimator closely aligned with the nominal $95$\% level.

Next, to investigate the robustness of the numerical results we consider a DGP model without common shocks. In particular, we use the same process as described in equation \eqref{eq:MC DGP}, but use a simpler process for the innovations such that $U_{it}=\varepsilon_{it}$, where $\varepsilon_{it} \sim N(0,1)$. The corresponding results for bias and RMSE are provided in Table \ref{tab:mc_bias_rmse_split_NOCS}, and those for coverage rates in Table \ref{tab:mc_coverage_split_NOCS}. 

\begin{table}[h]
\centering
\begin{threeparttable}
\caption{Monte Carlo performance of the FEQR  estimator (no common shocks): bias and RMSE}
\label{tab:mc_bias_rmse_split_NOCS}
\begin{tabular}{lcccccc}
\toprule
&\multicolumn{2}{c}{$\tau=0.25$}&\multicolumn{2}{c}{$\tau=0.50$}&\multicolumn{2}{c}{$\tau=0.75$}\\
\cmidrule(lr){2-7}
$(N,T)$ & Bias & RMSE  & Bias & RMSE   & Bias  & RMSE   \\
\midrule
\addlinespace[0.25em]
$(100,10)$  & 0.0128 & 0.0402 &  0.0002 &  0.0364 & -0.0126 & 0.0421 \\
$(100,25)$  & 0.0048 & 0.0248 &  0.0004 &  0.0228 & -0.0035 & 0.0252 \\
$(100,50)$  & 0.0024 & 0.0173 &  0.0003 &  0.0160 & -0.0027 & 0.0177 \\
$(100,100)$  & 0.0009 & 0.0122 &  -0.0004 &  0.0112 & -0.0015 & 0.0124 \\
$(250,10)$  & 0.0127 & 0.0284 &  0.0006 &  0.0239 & 0.0116 & 0.0277 \\
$(250,25)$  & 0.0043 & 0.0163 &  -0.0002 &  0.0140 & -0.0043 & 0.0159 \\
$(250,50)$  & 0.0021 &  0.0111 & 0.0001 & 0.0103 &  -0.0021 &  0.0112 \\
$(250,100)$  & 0.0011 & 0.0080 &  0.0001 &  0.0071 & -0.0010 & 0.0077 \\
$(500,10)$  & 0.0117 & 0.0209 & 0.0002 & 0.0167 &  -0.0116 &  0.0216 \\
$(500,25)$  & 0.0046 & 0.0119 & 0.0003 & 0.0104 &  -0.0039 &  0.0118 \\
$(500,50)$  & 0.0022 & 0.0081 & 0.0001 & 0.0072 &  -0.0019 & 0.0079\\
$(500,100)$  & 0.0010 & 0.0056 &  -0.0001 &  0.0050 & -0.0012 & 0.0055 \\
$(1000,10)$ & 0.0119 & 0.0174 & -0.0001 & 0.0119 & -0.0117 & 0.0170 \\
$(1000,25)$ & 0.0043 & 0.0088 & 0.0001 & 0.0071 & -0.0041 & 0.0090\\
$(1000,50)$ & 0.0020 & 0.0059 & 0.0001 & 0.0053 & -0.0019 & 0.0059\\
$(1000,100)$  & 0.0010 & 0.0041 &  -0.0001 &  0.0035 & -0.0009 & 0.0040 \\
\bottomrule
\end{tabular}
\begin{tablenotes}[flushleft]\footnotesize
\item \emph{Notes:} Entries report the Monte Carlo bias $\E[\hat\beta(\tau)]-\beta(\tau)$ and RMSE $\sqrt{\E[(\hat\beta(\tau)-\beta(\tau))^2]}$ of the Koenker-type FEQR  estimator at quantile index $\tau$. The true coefficient is $\beta(\tau)=\beta+\gamma q_\tau$. Results are based on $2,000$ Monte Carlo replications under the DGP in Section~\ref{sec:simulation}. 
\end{tablenotes}
\end{threeparttable}
\end{table}

Results in Table \ref{tab:mc_bias_rmse_split_NOCS} are similar to those in Table \ref{tab:mc_bias_rmse_split} with the difference that RMSE are smaller. This is an expected result since consistency of the estimator holds under a mild condition that $(\log N)^2 / T \to 0$ described in Proposition \ref{prop:consistency}.

Coverage rates are collected in Table \ref{tab:mc_coverage_split_NOCS}. Under no common shocks, results for the Standard method are substantially improved relative to the previous case. The table also shows that empirical coverages for the Robust case are close to the nominal 95\% coverage. It is also interesting to see that both the Robust method seems to improve more as $T$ increases, for a fixed $N$. These results confirm numerically the robustness of the Robust method.

\begin{table}[h]
\centering
\begin{threeparttable}
\caption{Monte Carlo inference accuracy (no common shocks): coverage of robust and standard covariance estimators}
\label{tab:mc_coverage_split_NOCS}
\begin{tabular}{lcccccc}
\toprule
 &\multicolumn{2}{c}{$\tau=0.25$}&\multicolumn{2}{c}{$\tau=0.50$}&\multicolumn{2}{c}{$\tau=0.75$}\\
\cmidrule(lr){2-7}
 $(N,T)$ & Robust & Standard & Robust & Standard & Robust & Standard \\
\midrule
\addlinespace[0.25em]
$(100,10)$  & 0.921 & 0.960 &  0.947 &  0.974 & 0.907 & 0.956 \\
$(100,25)$  & 0.953 & 0.962 &  0.959 &  0.971 & 0.942 & 0.951 \\
$(100,50)$  & 0.957 & 0.959 &  0.955 &  0.958 & 0.953 & 0.955 \\
$(100,100)$  & 0.956 & 0.963 &  0.962 &  0.960 & 0.951 & 0.955 \\
$(250,10)$  & 0.885 & 0.935 & 0.940 & 0.970 & 0.898  & 0.948  \\
$(250,25)$  & 0.939 & 0.954 & 0.963 & 0.974 & 0.938  & 0.954  \\
$(250,50)$  & 0.950 & 0.959 & 0.986 & 0.962 & 0.950 &  0.958 \\
$(250,100)$  & 0.947 & 0.948 &  0.957 &  0.959 & 0.952 & 0.957 \\
$(500,10)$  & 0.880 & 0.928 & 0.943 & 0.970 & 0.873 & 0.917 \\
$(500,25)$  & 0.928 & 0.945 & 0.958 & 0.966 & 0.919 & 0.941 \\
$(500,50)$  & 0.945 & 0.949 & 0.958 & 0.964 & 0.950 & 0.957 \\
$(500,100)$  & 0.955 & 0.957 &  0.959 &  0.961 & 0.957 & 0.962 \\
$(1000,10)$ & 0.818 & 0.881 & 0.941 & 0.968 & 0.819 & 0.890 \\
$(1000,25)$ & 0.914 & 0.934 & 0.958 & 0.970 & 0.904 & 0.921 \\
$(1000,50)$ & 0.931 & 0.937 & 0.947 & 0.953 & 0.931 & 0.945 \\
$(1000,100)$  & 0.939 & 0.940 &  0.961 &  0.966 & 0.941 & 0.949 \\
\bottomrule
\end{tabular}
\begin{tablenotes}[flushleft]\footnotesize
\item \emph{Notes:} Entries report empirical coverage probabilities of nominal 95\% confidence intervals for $\beta(\tau)$ constructed using (i) the proposed robust covariance estimator and (ii) the conventional sandwich covariance estimator that ignores cross-sectional dependence.  Results are based on $2,000$ Monte Carlo replications under the DGP in Section~\ref{sec:simulation}. Bandwidth: $h=1.06 \cdot \mathrm{sd}(\hat{\epsilon})\cdot T^{-1/5}$.
\end{tablenotes}
\end{threeparttable}
\end{table}

\medskip

\section{Conclusion}\label{sec:conclusion}

This paper develops an asymptotic theory for fixed-effects panel quantile regression with pervasive common shocks, a feature central to many economic and financial panels but largely absent from existing FEQR theory. By modelling outcomes through unit, time, and idiosyncratic latent components, we allow for general cross-sectional dependence while retaining the standard FEQR estimator. The key insight is that conditioning on common time shocks induces DGP-induced smoothing, which restores local differentiability of projected scores and enables Taylor expansions despite the non-smooth objective. This yields consistency and asymptotic normality under joint asymptotics with \(N,T \to \infty\), including regimes with \(T \ll N\).
Common shocks fundamentally alter the stochastic structure of FEQR: the estimator concentrates around a $\sqrt{T}$-order asymptotic component, with an asymptotic covariance matrix that differs from that in classical settings, rendering conventional FEQR variance estimator inconsistent.
We therefore develop inference procedures, including a novel robust variance estimator, that deliver valid confidence intervals and hypothesis tests in the presence of systemic time shocks while remaining applicable under the classical FEQR framework.

Several directions for future research appear promising. First, the framework could be extended to other panel models with non-smooth objective functions -- such as censored or threshold-type estimators -- where analogous forms of inherent smoothing may emerge in the presence of common shocks. Second, an important extension is to permit serially correlated time factors and to develop inference procedures that simultaneously accommodate both cross-sectional and temporal dependence in non-smooth panel settings. Finally, incorporating two-way fixed effects into panel quantile regression under full generality remains theoretically demanding and represents a particularly interesting open problem.

\medskip


\newpage

\section*{Appendix}\appendix

\subsection*{Notations}
Throughout the appendix, 
for $f:\mathcal{B}\to\mathbb{R}$, define $\mathbb{P}_T f=T^{-1}\sum_{t=1}^T f(B_t)$. 
For each $i\in\mathbb{N}$, let us write $$h_i(b)=f_i(0\mid B_t=b),\qquad k_i(b)=\mathbb{E}[f_i(0\mid X_{it},B_t)X_{it}\mid B_t=b]$$
for the conditional densities and
conditional density-weighted moments. 

\section{Proof of the Main Results}

\subsection{Proof of Proposition \ref{prop:consistency} and Proposition \ref{prop:consistency_m}}\label{sec:proof of prop:consistency}
    The proof for Proposition \ref{prop:consistency} follows the same arguments as in the proof of Theorem 3.1 in \cite{kato2012asymptotics} and is not repeated here--note that the dependence structure was only used in bounding their Equation (A.4) when Marcinkiewicz-Zygmund inequality is invoked with independence across $t=1,...,T$, which is satisfied under our setup following \eqref{eq:DGP}. Here, under Assumption \ref{assump:x_bounded_support}, one can apply Hoeffding's inequality in place of Marcinkiewicz-Zygmund and hence the desired results requires only a weaker condition of  $(\log N)/T =o(1)$.

    The proof of Proposition \ref{prop:consistency_m} is identical, except that, when bounding the second term in Equation (A.4), for each $1 \leq i \leq N$, we split $\Delta_{ni}$ into $M+1$ subseries, so that any two terms within the same subseries are at least $M+1$ periods apart and hence independent. Since $M$ is fixed, this splitting does not affect the rate conclusions, and the desired result follows.
\qed

\subsection{Proof of Theorem \ref{thm:main} and \ref{thm:main_m}}\label{sec:proof of thm:main}
 We first prove Theorem \ref{thm:main}.
 Observe that Lemma \ref{lem:QR_computational} show that the subgradients satisfy
   \begin{align*}
        &O_p(T^{-1}) = \sup_{1 \leq i \leq N} | \mathbb{H}_{Ni}^{(1)}(\hat{\alpha}_i, \hat{\beta}) |, \\
        &O_p(T^{-1}) = \norm{\mathbb{H}_{N}^{(2)}(\hat{\bm{\alpha}}, \hat{\beta})}.
        \end{align*}
Further,  Lemma \ref{lem:differentiability_expectations} implies the differentiability of $\tilde{\mathbb{H}}_{Ni}^{(1)}$ and $\tilde{\mathbb{H}}_{N}^{(2)}$, and it also shows the uniform boundedness of the second-order derivatives.
Hence, Taylor expansions on $\tilde{\mathbb{H}}_{N}^{(2)}(\hat{\bm{\alpha}}, \hat{\beta})$ yield
\begin{align*}
    O_p(T^{-1}) =& \mathbb{H}_{N}^{(2)}(\hat{\bm{\alpha}}, \hat{\beta}) \\
    =& \tilde{\mathbb{H}}_{N}^{(2)}(\hat{\bm{\alpha}}, \hat{\beta}) + (\mathbb{H}_{N}^{(2)}(\hat{\bm{\alpha}}, \hat{\beta}) - \tilde{\mathbb{H}}_{N}^{(2)}(\hat{\bm{\alpha}}, \hat{\beta})) \\
    =& \tilde{\mathbb{H}}_{N}^{(2)}(\bm{\alpha}_0, \beta_0) - \left( \frac{1}{NT} \sum_{i = 1}^N \sum_{t = 1}^T \E[f_i(0 \mid X_{it}, B_t)X_{it}X_{it}' \mid B_t] \right) (\hat{\beta} - \beta_0) \\
    &- \frac{1}{NT} \sum_{i = 1}^N (\hat{\alpha}_i - \alpha_{i0}) \sum_{t = 1}^T \E[f_i(0 \mid X_{it}, B_t)X_{it} \mid B_t] \\
    &+ O_p\left(\max_{1 \leq i \leq N} |\hat{\alpha}_i - \alpha_{i0}|^2 \vee \norm{\hat{\beta} - \beta_0}^2\right) \\
    &+ \left(\mathbb{H}_{N}^{(2)}(\hat{\bm{\alpha}}, \hat{\beta}) - \tilde{\mathbb{H}}_{N}^{(2)}(\hat{\bm{\alpha}}, \hat{\beta})\right).
    \numberthis \label{eq:H2_expansion}
\end{align*}
Also, by Taylor expansions on each  $\tilde{\mathbb{H}}_{Ni}^{(1)}(\hat{\alpha}_i, \hat{\beta})$, it holds  uniformly over $1 \leq i \leq N$
\begin{align*}
    O_p(T^{-1}) =& \mathbb{H}_{Ni}^{(1)}(\hat{\alpha}_i, \hat{\beta}) \\
    =& \tilde{\mathbb{H}}_{Ni}^{(1)}(\hat{\alpha}_i, \hat{\beta})  +  \left(\mathbb{H}_{Ni}^{(1)}(\hat{\alpha_{i}}, \hat{\beta}) - \tilde{\mathbb{H}}_{Ni}^{(1)}(\hat{\alpha}_i, \hat{\beta})\right) \\
    =& \tilde{\mathbb{H}}_{Ni}^{(1)}(\alpha_{i0}, \beta_0) - \left( \frac{1}{T} \sum_{i = 1}^T f_i(0 \mid B_t)\right) (\hat{\alpha}_i - \alpha_{i0}) \\ 
    &- \left( \frac{1}{T} \sum_{t = 1}^T \E[f_i(0 \mid X_{it}, B_t)X_{it}' \mid B_t] \right)(\hat{\beta} - \beta_0) \\
    &+ O_p\left(\max_{1 \leq i \leq N} |\hat{\alpha}_i - \alpha_{i0}|^2 \vee \norm{\hat{\beta} - \beta_0}^2\right) \\
    &+ \left(\mathbb{H}_{Ni}^{(1)}(\hat{\alpha_{i}}, \hat{\beta}) - \tilde{\mathbb{H}}_{Ni}^{(1)}(\hat{\alpha}_i, \hat{\beta})\right). 
\end{align*}
By rearranging the above, we have uniformly over $1 \leq i \leq N$,
\begin{align}
    \hat{\alpha}_i - \alpha_{i0} =& \left( \frac{1}{T} \sum_{t = 1}^T f_i(0 \mid B_t) \right)^{-1} \tilde{\mathbb{H}}_{Ni}^{(1)}(\alpha_{i0}, \beta_{i0})\nonumber \\
    &- \left( \frac{1}{T} \sum_{t = 1}^T f_i(0 \mid B_t) \right)^{-1} \left( \frac{1}{T} \sum_{t = 1}^T \E[f_i(0 \mid X_{it}, B_t)X_{it}' \mid B_t] \right)(\hat{\beta} - \beta_0) \nonumber\\
    &+ O_p\left(\max_{1 \leq i \leq N} |\hat{\alpha}_i - \alpha_{i0}|^2 \vee \norm{\hat{\beta} - \beta_0}^2\right)\nonumber\\
    &+ \left( \frac{1}{T} \sum_{t = 1}^T f_i(0 \mid B_t) \right)^{-1}\left(\mathbb{H}_{Ni}^{(1)}(\hat{\alpha_{i}}, \hat{\beta}) - \tilde{\mathbb{H}}_{Ni}^{(1)}(\hat{\alpha}_i, \hat{\beta})\right) \nonumber\\
    &+ O_p(T^{-1}).\label{eq:alpha_representation}
\end{align}
Recall that
$
    h_i(b) = f_i(0 \mid B_t = b)$ and  $k_i(b) = \E[f_i(0 \mid X_{it}, B_t)X_{it} \mid B_t = b]$.
Plugging \eqref{eq:alpha_representation} into \eqref{eq:H2_expansion},  we obtain 
\begin{align*}
    O_p(T^{-1}) =& \tilde{\mathbb{H}}_N^{(2)}(\alpha_0, \beta_0) - \frac{1}{N} \sum_{i = 1}^N (\mathbb{P}_T h_i)^{-1} (\mathbb{P}_T k_i) \tilde{\mathbb{H}}_{Ni}^{(1)}(\alpha_{i0}, \beta_0) \\
    &- \Gamma_N (\hat{\beta} - \beta_0) \\
    &+ O_p\left( \max_{1 \leq i \leq N} |\hat{\alpha}_i - \alpha_{i0}|^2 \vee \norm{\hat{\beta} - \beta_0}^2\right) \left( \frac{1}{N}\sum_{i = 1}^N \mathbb{P}_T k_i + 1 \right) \\
    &+ \frac{1}{N} \sum_{i = 1}^N (\mathbb{P}_T h_i)^{-1} (\mathbb{P}_T k_i) \left( \mathbb{H}^{(1)}_{Ni}(\hat{\alpha}_i, \hat{\beta}) - \tilde{\mathbb{H}}^{(1)}_{Ni}(\hat{\alpha}_i, \hat{\beta}) \right) \\
    &+ \mathbb{H}_N^{(2)}(\hat{\alpha}, \hat{\beta}) - \tilde{\mathbb{H}}_N^{(2)}( \hat{\alpha}, \hat{\beta}),
\end{align*}
where 
\begin{align*}
    \Gamma_N =& \frac{1}{NT} \sum_{i = 1}^N \sum_{t = 1}^T \E[f_i(0 \mid X_{it}, B_t) X_{it} X_{it}' \mid B_t] \\
    &- \frac{1}{N
    T} \sum_{i = 1}^N \left(\sum_{t = 1}^T \E[f_i(0 \mid X_{it}, B_t) X_{it} \mid B_t] \right)(\mathbb{P}_T h_i)^{-1} (\mathbb{P}_T k_i)'.
\end{align*}
Rearranging the above gives 
\begin{align*}
    \hat{\beta} - \beta_0 + o_p\left( \norm{\hat{\beta} - \beta_0}\right)
    =& \Gamma_N^{-1} \left( \tilde{\mathbb{H}}_N^{(2)}(\alpha_0, \beta_0) - \frac{1}{N} \sum_{i = 1}^N (\mathbb{P}_T h_i)^{-1} (\mathbb{P}_T k_i) \tilde{\mathbb{H}}_{Ni}^{(1)}(\alpha_{i0}, \beta_0)\right) \\
    &+ \Gamma_N^{-1} \left( \frac{1}{N} \sum_{i = 1}^N (\mathbb{P}_T h_i)^{-1} (\mathbb{P}_T k_i) \left( \mathbb{H}^{(1)}_{Ni}(\hat{\alpha}_i, \hat{\beta}) - \tilde{\mathbb{H}}^{(1)}_{Ni}(\hat{\alpha}_i, \hat{\beta}) \right)\right) \\
    &+ \Gamma_N^{-1} \left( \mathbb{H}_N^{(2)}(\hat{\alpha}, \hat{\beta}) - \tilde{\mathbb{H}}_N^{(2)}( \hat{\alpha}, \hat{\beta}) \right) \\
    &+ O_p\left( \max_{1 \leq i \leq N} |\hat{\alpha}_i - \alpha_{i0}|^2 \right) \\
    &+ O_p(T^{-1}). \numberthis \label{eq:beta_representation}
\end{align*}


By applying Lemmas \ref{lem:hajek_rs_1}, \ref{lem:hajek_rs_2}, \ref{lem:max_alpha_rate}, and \ref{lem:convergence_Gamma} in Section \ref{sec:auxiliary_lemmas} in the appendix,  the second to  the fourth terms on the RHS of \eqref{eq:beta_representation} are together of the order $o_p(T^{-1/2})$.
Hence \eqref{eq:beta_representation} becomes
\begin{align*}
    \hat{\beta} - \beta_0 =&\Gamma_N^{-1} \left( \tilde{\mathbb{H}}_N^{(2)}(\alpha_0, \beta_0) - \frac{1}{N} \sum_{i = 1}^N (\mathbb{P}_T h_i)^{-1} (\mathbb{P}_T k_i) \tilde{\mathbb{H}}_{Ni}^{(1)}(\alpha_{i0}, \beta_0)\right) +o_p\left( T^{-1/2}\right) \\
    =&\Gamma_N^{-1} \left( \tilde{\mathbb{H}}_N^{(2)}(\alpha_0, \beta_0) - \frac{1}{N} \sum_{i = 1}^N\gamma_i \tilde{\mathbb{H}}_{Ni}^{(1)}(\alpha_{i0}, \beta_0)\right) 
    + o_p\left( T^{-1/2}\right), \numberthis \label{eq:beta_bahadur}
\end{align*}
where
$\gamma_i = \E[f_i(0 \mid X_{i1})X_{i1}] / f_i(0)$. This yields an asymptotic linear representation of $\hat \beta- \beta_0$.
The desirable result follows directly from the representation of \eqref{eq:beta_bahadur}, Lemmas \ref{lem:clt} and  \ref{lem:convergence_Gamma}, and an application of the Slutsky's lemma.

We now prove Theorem \ref{thm:main_m}. First, note that the derivation up to Equation \eqref{eq:beta_representation} does not rely on serial independence of $\{B_t\}$. Moreover, the rate conclusions in Lemmas \ref{lem:hajek_rs_1}, \ref{lem:hajek_rs_2}, \ref{lem:max_alpha_rate}, and \ref{lem:convergence_Gamma} continue to hold when $\{B_t\}$ is stationary and $M$-dependent. Indeed, the relevant bounds can be obtained by splitting the sequence into finitely many subsequences of the form
\[
\{B_j,B_{j+(M+1)},B_{j+2(M+1)},\ldots\},
\qquad 1\leq j\leq M+1,
\]
whose elements are independent. Since $M$ is fixed, this decomposition keeps the rates unchanged. Therefore, Equation \eqref{eq:beta_bahadur} remains valid. The desired conclusion then follows from Lemma \ref{lem:clt_m} and Slutsky's lemma.

\qed


\subsection{Proof of Theorem \ref{thm:covariance_estimation} and \ref{thm:covariance_estimation_m}}\label{sec:proof of thm:covariance_estimation}

We first show Theorem \ref{thm:covariance_estimation}.

\noindent\textbf{(i)}. 
First, note that following the same steps as Proposition~3.1 in \cite{kato2012asymptotics}, it holds 
uniformly over $i = 1, \ldots, N$, that
\begin{align}
  &\frac{1}{T}\sum_{t=1}^T K_h(\hat\epsilon_{it}) = f_i(0) + o_p(1), \nonumber\\
  &\frac{1}{T}\sum_{t=1}^T K_h(\hat\epsilon_{it}) X_{it}
    = \mathbb{E}\!\left[f_i(0 \mid X_{i1}) X_{i1}\right] + o_p(1),\nonumber \\
  &\frac{1}{T}\sum_{t=1}^T K_h(\hat\epsilon_{it}) X_{it} X_{it}' 
    = \mathbb{E}\!\left[f_i(0 \mid X_{i1}) X_{i1} X_{i1}'\right] + o_p(1)\label{eq:covar_Kato_et_al_12}
\end{align}
Observe that
under the maintained assumptions, Proposition~\ref{prop:consistency} implies weak consistency of 
$(\hat\alpha_1,\ldots,\hat\alpha_N,\hat\beta')$, which allows one to replace 
$\hat\epsilon_{it}$ with its population counterpart up to an asymptotically negligible error. 
Moreover, for each fixed $i$, the summands in the three averages on the left-hand side are independent across $t = 1, \ldots, T$, as a consequence of the independence of $\{(B_t, U_{it})\}_{t=1}^T$ over $t$,
so these uniform consistency statements follow from the same Bousquet's inequality arguments used in \cite{kato2012asymptotics}.
These results imply that $\sup_{i\in \mathbb N}\|\hat \gamma_i-\gamma_i\|=o_p(1)$.

Now, let us define the oracle quantities of 
\begin{align*}
    m^0_{Nt} =&\frac{1}{N}\sum_{i=1}^N (\tau - \idf\{\epsilon_{it} \leq 0\})(X_{it} - \gamma_i), \quad 
    \bar{m}^0_{NT} =\frac{1}{T} \sum_{t = 1}^T m^0_{Nt}.
\end{align*}
and write $m^0_{Ntj}$, $\bar{m}^0_{NTj}$ for their respective $j$-th element for $j=1,...,p$. 
The difference between $m^0_{NT}$ and $\hat{m}_{NT}$ in \eqref{eq:mNt} is that we replace $\hat{\gamma}_i$ and $\hat \epsilon_{it}$ by the true $\gamma_i$ and $\epsilon_{it}$, respectively.
Notice that conditioning on $\{B_t\}_{t=1}^T$, for each $t=1,...,T$, the summands in $m_{Nt}^0$ and $\bar{m}^0_{NT}$ are independent over $i$. 
By Hoeffding's Inequality for conditional probabilities and Assumption \ref{assump:x_bounded_support}, we have for all $1 \leq t \leq T$, $1 \leq j \leq p$ and $s > 0$
\begin{align*}
    P\{ \left|m^0_{Ntj} - \E[m^0_{Ntj} \given B_t] \right| \geq s \given B_t\} \leq 2 C \exp\left( - 2 N s^2\right),
\end{align*}
where $C$ is a constant that does not depend on $t$, $N$, $T$ and $B_t$. 
This implies that, 
\begin{align*}
    P\{ \left|m^0_{Ntj} - \E[m^0_{Ntj} \given B_t] \right| \geq s \} \leq 2 C \exp\left( - 2 N s^2\right).
\end{align*}
By the uniform bound, we have 
\begin{align*}
    \max_{1 \leq t \leq T} \norm{m^0_{Nt} - \E[m^0_{Nt} \given B_t]} = O_p\left( \sqrt{\frac{\log T}{N}}\right).
\end{align*}
Define the following oracle estimator for $\Sigma$, $\tilde{\Sigma}_N = \frac{1}{T} \sum_{t = 1}^T m^0_{Nt} (m^0_{Nt})' - (\bar{m}^0_{NT})(\bar{m}^0_{NT})'$. 
By the assumption that $(\log T) / N \to 0$, we see that 
\begin{align*}
    \tilde{\Sigma}_N = \frac{1}{T} \sum_{t = 1}^T \E[m^0_{Nt} \given B_t] E[m^0_{Nt} \given B_t]' - \left( \frac{1}{T} \sum_{t = 1}^T \E[m^0_{Nt} \given B_t]\right)\left( \frac{1}{T} \sum_{t = 1}^T \E[m^0_{Nt} \given B_t]\right)' + o_p(1).
\end{align*}
An application of Markov's Inequality yields that for any $N$ and $\epsilon > 0$, we have 
\begin{align*}
    P \left\{\norm{\frac{1}{T} \sum_{t = 1}^T \E[m^0_{Nt} \given B_t] \E[m^0_{Nt} \given B_t]' - \E\left[ \E[m^0_{Nt} \given B_t]\E[m^0_{Nt} \given B_t]' \right]} > \epsilon\right\} \leq 
    \frac{D}{T\epsilon^2},
\end{align*}
where $D$ is a constant independent of $N$ and $T$. This shows that, as $N, T \to \infty$,
\begin{align*}
    \norm{\frac{1}{T} \sum_{t = 1}^T \E[m^0_{Nt} \given B_t] \E[m^0_{Nt} \given B_t]' - \E\left[ \E[m^0_{Nt} \given B_t]\E[m^0_{Nt} \given B_t]' \right]} = o_p(1).
\end{align*}
Similarly, as $N, T \to \infty$,
\begin{align*}
    \norm{\frac{1}{T} \sum_{t = 1}^T \E[m^0_{Nt} \given B_t] - \E[m^0_{Nt}]} = o_p(1).
\end{align*}
This shows that 
\begin{align*}
    \tilde{\Sigma}_N =& \E \left( \E[m^0_{Nt} \given B_t]\E[m^0_{Nt} \given B_t]' \right) - \E[m^0_{N1}]\E[m^0_{N1}]' + o_p(1) \\
    =& \var\left( \E[m^0_{N1} \given B_1]\right) + o_p(1).
\end{align*}
Recall from Assumption \eqref{assump:variance_bounded_gamma} that 
\begin{align*}
    \Sigma = \lim_{N \to \infty} \var\left( \E[m^0_{N1} \given B_1]\right).
\end{align*}
Therefore, we have the consistency of the oracle estimator $\|\tilde{\Sigma}_N - \Sigma\|= o_p(1)$.

 We now claim $\|\hat \Sigma_N-\tilde \Sigma_N\|=o_p(1)$.
Observe from \eqref{eq:mNt} and \eqref{eq:Sigma_hat_and_Gamma_hat} that 
\begin{align*}
    \hat{\Sigma}_N =& \frac{1}{T} \sum_{t = 1}^T \hat{m}_{Nt} (\hat{m}_{Nt})' - \bar{m}_{NT} (\bar{m}_{NT})'. 
\end{align*}
Without loss of generality,
suppose $\alpha_{i0} = 0$ for all $i \in \N$ and $\beta_0 = 0$. Write 
\begin{align*}
    \hat{z}_{it} = \left(\tau - \idf\{\epsilon_{it} - \hat{\alpha}_i -X_{it}'\hat{\beta} \leq 0\}\right)(X_{it} - \hat \gamma_i ), \quad 
    z^0_{it} = \left( \tau - \idf\{\epsilon_{it} \leq 0\}\right)(X_{it} - \gamma_i).
\end{align*}
Observe that 
\begin{align*}
    \hat{\Sigma}_N &= \frac{1}{N^2} \sum_{i = 1}^N \sum_{j = 1}^N \left(\frac{1}{T}\sum_{t = 1}^T \hat{z}_{it} (\hat{z}_{jt})' \right)
    - \left( \frac{1}{NT} \sum_{i = 1}^N \sum_{t = 1}^T \hat{z}_{it}\right)
    \left( \frac{1}{NT} \sum_{i = 1}^N \sum_{t = 1}^T \hat{z}_{it}\right)', \\
    \tilde{\Sigma}_N &= 
    \frac{1}{N^2} \sum_{i = 1}^N \sum_{j = 1}^N \left(\frac{1}{T}\sum_{t = 1}^T z^0_{it} (z^0_{jt})' \right)
    - \left( \frac{1}{NT} \sum_{i = 1}^N \sum_{t = 1}^T z^0_{it}\right)
    \left( \frac{1}{NT} \sum_{i = 1}^N \sum_{t = 1}^T z^0_{it}\right)'.
\end{align*}
If follows that 
\begin{align*}
 \norm{\hat{\Sigma}_N - \tilde{\Sigma}_N}
    \leq& \frac{1}{N^2} \sum_{i = 1}^N \sum_{j = 1}^N \norm{\frac{1}{T}\sum_{t = 1}^T \hat{z}_{it} (\hat{z}_{jt})' - \frac{1}{T}\sum_{t = 1}^T z^0_{it} (z^0_{jt})'} \\
    &+ 2\left( \max_{1 \leq i \leq N} \norm{ \frac{1}{T} \sum_{t = 1}^T z^0_{it}} + \frac{1}{N^2}\sum_{i = 1}^N \sum_{j = 1}^N \norm{\frac{1}{T} \sum_{t = 1}^T (\hat{z}_{it} -  z^0_{it})}\right)  \\
    &\quad\times \frac{1}{N^2} \sum_{i = 1}^N \sum_{j = 1}^N\norm{\frac{1}{T} \sum_{t = 1}^T (\hat{z}_{it} - z^0_{it})}.
\end{align*}
Therefore, to show that $\norm{\hat{\Sigma}_N - \Sigma} = o_p(1)$, by Markov's inequality, it suffices to show that 
\begin{align*}
    &\max_{1 \leq i  \leq N} \E \norm{\frac{1}{T} \sum_{t = 1}^T (\hat{z}_{it} - z^0_{it})} = o(1),\quad
      \max_{1 \leq i\leq j \leq N} \E\norm{\frac{1}{T}\sum_{t = 1}^T \left(\hat{z}_{it} (\hat{z}_{jt})' -  z^0_{it} (z^0_{jt})' \right)} = o(1) 
\end{align*}
 Let us show the first statement as the second can be shown analogously.  Recall that $\max_{1 \leq i \leq N} \norm{\hat{\gamma}_i -\gamma_i} = o_p(1)$ from \eqref{eq:covar_Kato_et_al_12} and $\sup_{(i,t)\in\mathbb N^2 }\|X_{it}\|<\infty$ and $\sup_{i\in \mathbb N}\|\gamma_i\|<\infty$ from Assumptions  \ref{assump:x_bounded_support} and \ref{assump:variance_bounded_gamma}, it suffices to show
\begin{align*}
&\max_{1 \leq i  \leq N} \E \norm{\frac{1}{T} \sum_{t = 1}^T  (\idf\{\hat \epsilon_{it}\le 0\}-\idf\{\epsilon_{it}\le 0\})X_{it}} = o(1),\\
&\max_{1 \leq i  \leq N} \E \left\|\frac{1}{T} \sum_{t = 1}^T  (\idf\{\hat \epsilon_{it}\le 0\}-\idf\{\epsilon_{it}\le 0\})\gamma_i\right\| = o(1).
\end{align*}
We focus on the first term, since the argument for the second is analogous.
Observe that as $\hat \epsilon_{it}=\epsilon_{it}-\hat \alpha_i - X_{it}'\hat \beta$, we have for each $i=1,...,N$, it holds that 
\begin{align}
&\left\|\frac{1}{T} \sum_{t = 1}^T  (\idf\{\hat \epsilon_{it}\le 0\}-\idf\{\epsilon_{it}\le 0\}) X_{it}\right\|\nonumber\\
\le& \left\|\frac{1}{T} \sum_{t = 1}^T  (\idf\{ \hat\epsilon_{it}\le 0\}-\E[\idf\{\hat \epsilon_{it}\le 0\}|X_{it}])X_{it}\right\|\nonumber\\
&+  \left\|\frac{1}{T} \sum_{t = 1}^T  (\E[\idf\{\hat \epsilon_{it}\le 0\}|X_{it}]-\E[\idf\{ \epsilon_{it}\le 0\}|X_{it}])X_{it}\right\|\nonumber\\
&+  \left\|\frac{1}{T} \sum_{t = 1}^T  (\E[\idf\{ \epsilon_{it}\le 0\}|X_{it}]-\idf\{ \epsilon_{it}\le 0\})X_{it}\right\|\nonumber\\
\le&2\sup_{(a,b')\in \R^{1+p}}\left\|\frac{1}{T} \sum_{t = 1}^T  (\idf\{ \epsilon_{it}-X_{it}'b\le a\}-\E[\idf\{ \epsilon_{it}-X_{it}'b\le a\}|X_{it}])X_{it}\right\|+o_p(1),\label{eq:indicator_bound}
\end{align}
where the the second term in \eqref{eq:indicator_bound} is $o_p(1)$ following Proposition \ref{prop:consistency} and properties of $f_i(\epsilon|x,b)$ from Assumption \ref{assump:bounded_pdf_derivative}.  To control the first term in \eqref{eq:indicator_bound}, note as they all share the same $i$, $\epsilon_{it}$ are independent over $t$.
Further, note that the class  \[\{(x_1,...,x_p,\epsilon)\mapsto\idf\{\epsilon -(x_1,...,x_p)b\le a\} \cdot x_j:(a,b')\in \R^{1+p}, j=1,...,p\}\] is a VC-subgraph\footnote{The formal definition can be found in Section 2.6.2 of \cite{VanDerVaarWellner1996EmpiricalProcess}.
} class following Lemmas 2.6.15 and 2.6.18 in \cite{VanDerVaarWellner1996EmpiricalProcess}.
By Lemma \ref{lem:conditional_exp} and Theorem 2.14.1 in \cite{VanDerVaarWellner1996EmpiricalProcess},
there exists a $C>0$ independent of $N,T$ and $i$ such that for all $i=1,...,N$,  
\[\E\left[\sup_{(a,b')\in \R^{1+p}}\left\|\frac{1}{T} \sum_{t = 1}^T  (\idf\{ \epsilon_{it}-X_{it}'b\le a\}-\E[\idf\{ \epsilon_{it}-X_{it}'b\le a\}|X_{it}])X_{it}\right\|\right]\le C\sqrt{\frac{1}{T}}.\]
This implies that the first term in \eqref{eq:indicator_bound} is $o_p(1)$. 
Collecting these results, we now have
 $\|\hat \Sigma_N-\tilde \Sigma_N\|=o_p(1)$ and therefore $\|\hat \Sigma_N-\Sigma\|=o_p(1)$.

On the other hand, write $\tilde{\Gamma}_N = \frac{1}{N} \sum_{i = 1}^N \E[f_i(0 \given X_{i1})X_{i1}(X_{i1} - \gamma_i)'] = \Gamma + o(1)$.
By  some algebra, we have
\begin{align*}
    \norm{\hat{\Gamma}_N - \tilde{\Gamma}_N} \leq& \sup_{1 \leq i \leq N} \norm{ \frac{1}{T}\sum_{t=1}^T K_h(\hat\epsilon_{it}) X_{it} X_{it}' - \mathbb{E}\!\left[f_i(0 \mid X_{i1}) X_{i1} X_{i1}'\right]} \\
    &+ 2\left( C_f \sup_{x \in \mathcal{X}} \norm{x} + \sup_{1 \leq i \leq N} \norm{\frac{1}{T} \sum_{t = 1}^T K_h(\hat\epsilon_{it}) X_{it} - E[f_i(0 \given X_{i1})X_{i1}]} \right) \sup_{1 \leq i \leq N} \norm{\hat{\gamma}_i - \gamma_{i}} \\
    &+ 2\left( \sup_{1 \leq i \leq N} \norm{\gamma_i} + \sup_{1 \leq i \leq N}\norm{\hat{\gamma}_i - \gamma_i}\right)\sup_{1 \leq i \leq N} \norm{\frac{1}{T} \sum_{t = 1}^T K_h(\hat\epsilon_{it}) X_{it} - E[f_i(0 \given X_{i1})X_{i1}]}.
\end{align*}
By collecting the results from \eqref{eq:covar_Kato_et_al_12}, the right hand side is $o_p(1)$ and hence the proof of this case is concluded.

\medskip
\noindent\textbf{(ii)}.
 First, observe that the results in \eqref{eq:covar_Kato_et_al_12} remain valid in this scenario.
Further, notice that by the independent across $i,t$-pairs, following a variance calculation, the average over cross-product terms satisfies
\begin{align*}
    \frac{1}{N^2T} \sum_{t = 1}^T \sum_{i=1}^N\sum_{i'\ne i}^N   z^0_{it}  (z^0_{i't})'=O_p\left(\frac{1}{\sqrt{N^2T}}\right),
\end{align*}
where $z^0_{it} = ( \tau - \idf\{\epsilon_{it} \leq 0\})(X_{it} - \gamma_i)$.
Under the maintained asymptotic regime, these components remain $o_p(1)$ even after multiplication by $N$. The remainder of the proof follows from standard arguments using the weak law of large numbers and is therefore omitted.

Next we prove Theorem \ref{thm:covariance_estimation_m}.
First note that 
$M$-dependence does not affect the uniform validity of the results in \eqref{eq:covar_Kato_et_al_12} over $1 \leq i \leq N$.
Hence, by the same argument as in the proof of Theorem \ref{thm:covariance_estimation}, we obtain
$\hat{\Gamma}_N \pto \Gamma$.
Moreover, for each $0 \leq l \leq M$, 
the argument used to establish $\norm{\hat{\Sigma}_N - \Sigma} = o_p(1)$ in the proof of Theorem \ref{thm:covariance_estimation} 
also yields $\norm{\hat{\Sigma}^l_N - \Sigma^l_N} = o_p(1)$.
It follows that that $\hat{\Lambda}_N \pto \Lambda$.
This proves the desired result.
\qed

\section{Auxiliary Lemmas}\label{sec:auxiliary_lemmas}
Throughout this section, for $q\in[1,\infty)$, define $\|f\|_{Q,q}=(Q|f|^q)^{1/q}$. 
For a non-empty set $T$ and a function $f:T\to\mathbb{R}$, let $\|f\|_{T}=\sup_{t\in T}|f(t)|$. 
For a pseudometric space $(T,d)$, let $N(T,d,\varepsilon)$ denote the $\varepsilon$-covering number of $(T,d)$. 
For a class of functions $\mathcal F\ni f:\mathcal X\to \R$, we say $F:\mathcal X\to \R_+$ is an envelope for $\mathcal F$ if it is measurable and $\sup_{f\in\mathcal F}|f(x)|\le F(x)$ for all $x\in \mathcal X$.

	\begin{lemma}[Uniform covering for conditional expectations] \label{lem:conditional_exp}
	Let $\mathcal{F}$ be a class of functions $f:\mathcal{X}\times \mathcal Y\to \R$ with envelopes $F$ and $R$ a fixed probability measure on $\mathcal Y$. For a given $f\in \mathcal{F}$, let $\overline f:\mathcal{X}\to \R$ be $\overline f=\int f(x,y)dR(y)$. Set $\overline{\mathcal{F}}=\{\overline f:f\in\mathcal{F}\}$. Note that $\overline{F}$ is an envelope of $\overline{\mathcal{F}}$. 
	Then, for any $r,s\ge 1$, $\varepsilon\in(0,1]$,
	\begin{align*}
\sup_{Q}N(\overline{\mathcal{F}},\|\cdot\|_{Q,r},2\varepsilon\|\overline F\|_{Q,r})\le 	\sup_{Q}N(\mathcal{F},\|\cdot\|_{Q\times R,s},\varepsilon^r\|F\|_{Q \times R,s}),
	\end{align*}
	where $\sup_Q$ are taken over all distributions on $\mathcal X$.
\end{lemma}
\begin{proof}
    This is a direct consequence of Lemma A.2. in \cite{ghosal2000testing}.
\end{proof}


\begin{lemma}[Computational property of quantile regression]\label{lem:QR_computational}
    Suppose Assumption \ref{assump:x_bounded_support} and \ref{assump:pdf_exists} hold, then we have
        \begin{align*}
        &O_p(T^{-1}) = \sup_{1 \leq i \leq N} | \mathbb{H}_{Ni}^{(1)}(\hat{\alpha}_i, \hat{\beta}) | \\
        &O_p(T^{-1}) = \norm{\mathbb{H}_{N}^{(2)}(\hat{\bm{\alpha}}, \hat{\beta})}.
        \end{align*}
   
\begin{proof}
This result is a consequence of the well-known computational property of quantile regression. We include a proof for completeness. 

    Write 
    \begin{align*}
        Q_N(\bm{\alpha}, \beta) = \frac{1}{NT}\sum_{i = 1}^N \sum_{t = 1}^T \rho_\tau(Y_{it} - \alpha_i - X_{it}'\beta).
    \end{align*}
    The directional derivative of $Q_N$ with 
    respect to $\beta$ at $(\hat{\bm{\alpha}}, \hat{\beta})$ is given by 
    \begin{align*}
        D_\beta Q_N(w) \coloneq -\frac{1}{NT} \sum_{(i, t)} (X_{it}'w) \psi(Y_{it} - \hat{\alpha}_i - X_{it}'\hat{\beta}, -X_{it}'w),
    \end{align*}
    where 
    \begin{align*}
        \psi(u, v) = 
        \begin{cases}
            \tau - \idf\{u < 0\} & \text{ if } u \neq 0, \\
            \tau - \idf\{v < 0\} & \text{ if } u = 0.
        \end{cases}
    \end{align*}
    Let $h$ be the subset of $\{(i, t): 1 \leq i \leq N,; 1 \leq t \leq T\}$ such that $\hat{Y}_{it} = \hat{\alpha}_i + X_{it}'\hat{\beta}$.
    Define, $h_i = h \cap \{(j, t): j = i, 1 \leq t \leq T\}$. 
    Since $(\hat{\bm{\alpha}}, \hat{\beta})$ minimizes $Q_N$, we know that, for any $\norm{w} = 1$, 
    \begin{align*}
        D_\beta Q_N(w) \geq 0, \quad D_\beta Q_N(-w) \geq 0,
    \end{align*}
    which implies 
    \begin{align*}
        \frac{1}{NT} \sum_{h} (-X_{it}'w) (\tau - \idf\{X_{it}'w < 0\}) 
        \leq& \frac{w'}{NT} \sum_{h^c} X_{it}(\tau - \idf\{ Y_{it} < \hat{\alpha}_i + X_{it}'\hat{\beta}\}) \\
        \leq& \frac{1}{NT} \sum_{h} (-X_{it}'w) (\tau - \idf\{-X_{it}'w < 0\}).
    \end{align*}
  Notice that for each $i$, $|h_i|$ is at most $p + 1$ with probability 1. This is because the conditional law of $\epsilon \given X_i$ admits a p.d.f. $P_{X_i}$-almost every $x$ by Assumption \ref{assump:pdf_exists}.
  
  We conclude that,
    uniformly over $\norm{w} = 1$, w.p.1,
    \begin{align*}
        \left| \frac{w'}{NT} \sum_{h^c} X_{it}(\tau - \idf\{ Y_{it} < \hat{\alpha}_i + X_{it}'\hat{\beta}\})\right| \leq \frac{1}{NT} \left( \sup_{x \in \mathcal{X}} \norm{x}_\infty \right)N(p + 1) = O\left(\frac{1}{T}\right), \\
        \left| \frac{w'}{NT} \sum_{(i, t)} X_{it}(\tau - \idf\{ Y_{it} \leq \hat{\alpha}_i + X_{it}'\hat{\beta}\})\right| \leq \frac{2}{NT} \left( \sup_{x \in \mathcal{X}} \norm{x}_\infty \right)N(p + 1) = O\left(\frac{1}{T}\right),
    \end{align*}
   Hence we have,  
   $
        \norm{\mathbb{H}_{N}^{(2)}(\hat{\bm{\alpha}}, \hat{\beta})} = O_p(T^{-1}).
    $
    
    The proof for $\sup_{1 \leq i \leq N} | \mathbb{H}_{Ni}^{(1)}(\hat{\alpha}_i, \hat{\beta}) |$
    is analogous--using this argument, we have for all $1 \leq i \leq N$, with probability 1,
    \begin{align*}
        \left| \frac{1}{T}\sum_{h_i^c} (\tau - \idf\{Y_{it} < \hat{\alpha}_i + X_{it}'\hat{\beta}\})\right|
        \leq \frac{1}{T}(p + 1), \\
        \left| \frac{1}{T}\sum_{t = 1}^T (\tau - \idf\{Y_{it} \leq \hat{\alpha}_i + X_{it}'\hat{\beta}\})\right|
        \leq \frac{2}{T}(p + 1).
    \end{align*}
    which implies 
$  
        \sup_{1 \leq i \leq N} | \mathbb{H}_{Ni}^{(1)}(\hat{\alpha}_i, \hat{\beta}) | = O_p(T^{-1}).
   $
\end{proof}
\end{lemma}

\begin{lemma}[Projection differentiability]\label{lem:differentiability_expectations}
Under the setup of Section \ref{sec:model_estimation},
    suppose  Assumption \ref{assump:bounded_pdf_derivative} holds, then
    $ \tilde{\mathbb{H}}_{N}^{(2)}(\bm{\alpha}, \beta)$
    is twice continuously differentiable in $(\bm{\alpha}, \beta)$ almost surely. In addition, the first-order partial derivatives
    are given by 
    \begin{align*}
        \vpd{\tilde{\mathbb{H}}_{N}^{(2)}(\bm{\alpha}, \beta)}{\alpha_i} &= -\frac{1}{T} \sum_{t = 1}^T \E[f_i( \alpha_i - \alpha_{i0} + X_{it}'(\beta - \beta_0) \mid X_{it}, B_t)X_{it} \mid B_t], \\
        \vpd{\tilde{\mathbb{H}}_{N}^{(2)}(\bm{\alpha}, \beta)}{\beta} 
        &= -\frac{1}{NT} \sum_{i = 1}^N \sum_{t = 1}^T \E[f_i( \alpha_i - \alpha_{i0} + X_{it}'(\beta - \beta_0) \mid X_{it}, B_t)X_{it} X_{it}'\mid B_t].
    \end{align*}
    Further, the first- and second-order derivatives are uniformly bounded across $(\bm{\alpha}, \beta)$. The bounds of the derivatives are independent of $N$. Similarly, $\tilde{\mathbb{H}}_{Ni}^{(1)}(\alpha_i, \beta)$ are also twice continuously differentiable in $(\alpha_i, \beta)$ almost surely. The first-order partial derivatives are given by 
    \begin{align*}
        \vpd{\tilde{\mathbb{H}}^{(1)}_{Ni}(\alpha_i, \beta)}{\alpha_i} 
        &= -\frac{1}{T} \sum_{t = 1}^T f_i(\alpha_i - \alpha_{i0} + X_{it}'(\beta - \beta_0) \given B_t) \\
        \vpd{\tilde{\mathbb{H}}^{(1)}_{Ni}(\alpha_i, \beta)}{\beta} 
        &= -\frac{1}{T} \sum_{t = 1}^T \E[f_i(\alpha_i - \alpha_{i0} + X_{it}'(\beta - \beta_0) \given X_{it}, B_t)X_{it} \given B_t].
    \end{align*}
    Moreover, the first- and second-order derivatives are uniformly bounded across $(\alpha_i, \beta)$ and $i \in \N$.
\begin{proof}
    It suffices to prove that the partial derivatives with respect to 
    $(\alpha_i, \beta)$ of 
    \begin{align}
        \E[\idf\{Y_{it} \leq \alpha_i + X_{it}'\beta\})X_{it} \given B_t] \label{eq:conditional_exp}
    \end{align}
    exist and are twice continuously differentiable almost surely for all $i \in \N$. Observe that 
    \begin{align*}
        &\E[\idf\{Y_{it} \leq \alpha_i + X_{it}'\beta\})X_{it} \given B_t] \\
        =& \E[F_i(\alpha_i - \alpha_{i0} + X_{it}'(\beta - \beta_0) \mid X_{it}, B_t)X_{it} \given B_t]
    \end{align*}
    By Assumption \ref{assump:x_bounded_support} and \ref{assump:bounded_pdf_derivative}, with the dominated convergence theorem for conditional expectations, the above display is differentiable in $\alpha_i$ almost surely, with derivative 
    \begin{align*}
        &\vpd{\E[F_i(\alpha_i - \alpha_{i0} + X_{it}'(\beta - \beta_0) \mid X_{it}, B_t)X_{it} \given B_t]}{\alpha_i} \\
        =& \E[f_i( \alpha_i - \alpha_{i0} + X_{it}'(\beta - \beta_0) \mid X_{it}, B_t)X_{it} \mid B_t].
    \end{align*}
    The above is continuous in $(\alpha_i, \beta)$ by the dominated convergence theorem. The same holds 
    for $\beta$---the partial derivative exists and is continuous in $(\alpha_i, \beta)$ with the form 
    \begin{align*}
        &\vpd{\E[F_i(\alpha_i - \alpha_{i0} + X_{it}'(\beta - \beta_0) \mid X_{it}, B_t)X_{it} \given B_t]}{\beta} \\
        =& \E[f_i( \alpha_i - \alpha_{i0} + X_{it}'(\beta - \beta_0) \mid X_{it}, B_t)X_{it} X_{it}'\mid B_t].
    \end{align*}
    This proves  the continuous differentiability of \eqref{eq:conditional_exp}. By applying the same arguments on the first order partial derivatives, one reaches the conclusion that 
    \eqref{eq:conditional_exp}
    is twice continuously differentiable in $(\bm{\alpha}, \beta)$.
    The boundedness of the derivatives comes from Assumptions \ref{assump:x_bounded_support} and \ref{assump:bounded_pdf_derivative}.

    The proof for $\tilde{H}_{Ni}^{(1)}(\alpha_i, \beta)$ is analogous and
    thus omitted.
\end{proof}
\end{lemma}

\begin{lemma}[Projection error bound (i)] \label{lem:hajek_rs_1}
Under the setup of Section \ref{sec:model_estimation},
    suppose Assumptions \ref{assump:x_bounded_support}--\ref{assump:Jacobian} hold.
    Take $\delta_N \to 0$ such that 
    \begin{align*}
        \max_{1 \leq i \leq N} |\hat{\alpha}_i - \alpha_{i0}| \vee \norm{\hat{\beta} - \beta_0} = O_p(\delta_N).
    \end{align*}
    We have 
    \begin{align*}
       \left\| \frac{1}{N} \sum_{i = 1}^N (\mathbb{P}_T h_i)^{-1} (\mathbb{P}_T k_i) \left( \mathbb{H}^{(1)}_{Ni}(\hat{\alpha}_i, \hat{\beta}) - \tilde{\mathbb{H}}^{(1)}_{Ni}(\hat{\alpha}_i, \hat{\beta}) \right)\right\| = o_p(T^{-1/2}) + O_p\left( d_N \right).
    \end{align*}
    where $d_N := T^{-1} |\log \delta_N| \vee T^{-1/2} \delta_N^{1/2} |\log \delta_N|^{1/2}$.
 Hence, in conjunction with Proposition \ref{prop:consistency}, the preceding display is of order $o_p(T^{-1/2})$.
    \begin{proof}
        First let us define
        \begin{align}
            \psi(x, \epsilon;\delta_\alpha, \delta_\beta) \coloneq
            & \idf\{\epsilon - x'\delta_\beta \leq \delta_\alpha\} - \idf\{ \epsilon \leq 0\},\nonumber \\
            \tilde{\psi}_i(b; \delta_\alpha, \delta_\beta)
            \coloneq &\E_{i}[\idf\{\epsilon_{i1} - X_{i1}'\delta_\beta \leq \delta_\alpha\} - \idf\{ \epsilon_{i1} \leq 0\} \mid B_1 = b] \nonumber\\
            =& \int_{\mathcal{X} \times \mathcal{E}} (\idf\{\epsilon - x'\delta_\beta \leq \delta_\alpha\} - \idf\{ \epsilon \leq 0 \} ) \dd F_i(x, \epsilon \given b), \nonumber \\
            \phi_i(x, \epsilon, b; \delta_\alpha, \delta_\beta) \coloneq & \psi(x, \epsilon; \delta_\alpha, \delta_\beta) - \tilde{\psi}_i(b; \delta_\alpha, \delta_\beta).\label{eq:def_psi_phi}
        \end{align}
        Also denote $\hat{\delta}_{\alpha i} = \hat{\alpha}_i - \alpha_{i0}$ and $\hat{\delta}_{\beta} = \hat{\beta} - \beta_0$. Observe that 
        \begin{align*}
            &\frac{\sqrt{T}}{N} \sum_{i = 1}^N (\mathbb{P}_T h_i)^{-1} (\mathbb{P}_T k_i) \left( \mathbb{H}^{(1)}_{Ni}(\hat{\alpha}_i, \hat{\beta}) - \tilde{\mathbb{H}}^{(1)}_{Ni}(\hat{\alpha}_i, \hat{\beta}) \right) \\
            =&\frac{1}{N} \sum_{i = 1}^N  (\mathbb{P}_T h_i)^{-1} (\mathbb{P}_T k_i) \frac{1}{\sqrt{T}} \sum_{t = 1}^T \left( -\idf\{\epsilon_{it} \leq 0\} + \E_i[\idf\{\epsilon_{it} \leq 0\}\given B_t ] \right) \\
            &+ \frac{1}{N} \sum_{i = 1}^N  (\mathbb{P}_T h_i)^{-1} (\mathbb{P}_T k_i) \frac{1}{\sqrt{T}} \sum_{t = 1}^T \phi_i(X_{it}, \epsilon_{it}, B_t; \hat{\delta}_\alpha, \hat{\delta}_\beta) \numberthis \label{eq:20260050}
        \end{align*}
        Note that, conditional on $\bm{B}_T$, the terms are independent across $i$. By the law of total variance, the variance of the first term can be written as
        \begin{align*}
            \frac{1}{N^2} \sum_{i = 1}^N \E\left[ (\mathbb{P}_T h_i)^{-2} (\mathbb{P}_T k_i)'\left( \frac{1}{T} \sum_{t = 1}^T \var(\idf\{\epsilon_{it} \leq 0\}\mid B_t)\right) (\mathbb{P}_T k_i)\right]. 
        \end{align*}
        Notice that, as $N\to\infty$, the norm  of this variance expression is bounded by
        \begin{align*}
            \frac{1}{N^2} \sum_{i = 1}^N \E \norm{ (\mathbb{P}_T h_i)^{-1} \mathbb{P}_T k_i }^2 &\leq \frac{1}{N} \left( \inf_{i \in \N, b \in \mathcal{B}} f_i(0 \given b)\right)^{-1} \cdot \sup_{i \in \N, x \in \mathcal{X}, b \in \mathcal{B}} f_i(0\given x, b) \cdot \sup_{x \in \mathcal{X}} |x|.
        \end{align*}
        By Assumption
        \ref{assump:x_bounded_support} and \ref{assump:bounded_pdf_derivative},
        the above is $O(1/N)$.
        Therefore, by Markov's Inequality, the first term in \eqref{eq:20260050} is $o_p(1)$. The norm of the second term in \eqref{eq:20260050} is bounded above by 
        \begin{align*}
            &\left( \sup_{1 \leq i \leq N} \norm{\gamma_i} + \sup_{1 \leq i \leq N} \norm{(\mathbb{P}_T h_i)^{-1} \mathbb{P}_T k_i - \gamma_i}\right) \times \\
            &\frac{1}{N}\sum_{i = 1}^N \left|\frac{1}{\sqrt{T}} \sum_{t = 1}^T \phi_i(X_{it}, \epsilon_{it}, B_t; \hat{\delta}_\alpha, \hat{\delta}_\beta)\right|.
            \numberthis
            \label{202601301702}
        \end{align*}
        By Lemma \ref{lemma:pdf_gc} and Assumptions \ref{assump:bounded_pdf_derivative} and \ref{assump:variance_bounded_gamma}, the two terms in the round brackets in \eqref{202601301702} is $O_p(1)$. Note that  we  used the fact that
      $
         \E[f_i(0 \mid X_{i1},B_1)X_{i1}]
            =
            \E[f_i(0 \mid X_{i1})X_{i1}],
      $  
        as  $\E[f_i(0 \given X_{i1}, B_1) \given X_{i1}] = f_i(0 \given X_{i1})$ following the law of iterated expectations.

        Next, we show that the second term in \eqref{202601301702} is $O_p(\sqrt{T} d_N)$.
        By Markov's Inequality, it suffices to show 
        \begin{align}
            \max_{1 \leq i \leq N} \E \left| \sum_{t=1}^T \phi_i(X_{it}, \epsilon_{it}, B_t; \hat{\delta}_\alpha, \hat{\delta}_\beta) \right| = O_p(d_N T).
            \label{eq:2026-0131-1355}
        \end{align}
        We follow the proof strategy in Step 2 in the proof of Theorem 3.2 with appropriate modifications.
        Define the following function classes
        \begin{align*}
            \Psi &\coloneq \{\psi(\cdot, \cdot; \delta_\alpha, \delta_\beta): \mathcal{X} \times \mathcal{E} \to \R: |\delta_\alpha| \vee |\delta_\beta| \in \R\} ,\\
            \Psi_\delta &\coloneq \{ \psi(\cdot, \cdot; \delta_\alpha, \delta_\beta): \mathcal{X} \times \mathcal{E} \to \R: |\delta_\alpha| \vee |\delta_\beta| < \delta\},
        \end{align*}
        recall that $\psi$ is defined in \eqref{eq:def_psi_phi}.
        Note that all functions in $\Psi$ is bounded by $2$.
        By Lemmas 2.6.15, 2.6.16, and 2.6.18 of \cite{VanDerVaarWellner1996EmpiricalProcess},
        there exists constants $A \geq 3 \sqrt{e}$ 
        and $v$ depending only on $p = \dim(\mathcal{X})$ such that 
        \begin{align*}
            N(\Psi, \|\cdot\|_{Q,2}, 2\epsilon) \leq \left( \frac{A}{\epsilon}\right)^v.
        \end{align*}
        for every probability measure $Q_{x, \epsilon}$ on $\mathcal{X} \times \mathcal{E}$. 
        For each $i \in \mathbb{N}$, define 
        \begin{align*}
            \tilde{\Psi}_i &\coloneq \{\tilde{\psi}_i(\cdot; \delta_\alpha, \delta_\beta): \mathcal{B} \to \R: |\delta_\alpha| \vee |\delta_\beta| \in \R\}, \\
            \tilde{\Psi}_{i, \delta} &\coloneq \{\tilde{\psi}_i(\cdot; \delta_\alpha, \delta_\beta): \mathcal{B} \to \R: |\delta_\alpha| \vee |\delta_\beta| < \delta\}, \\
            \Phi_i &\coloneq \{\phi_i(\cdot, \cdot, \cdot; \delta_\alpha, \delta_\beta): \mathcal{X} \times \mathcal{E} \times \mathcal{B} \to \R: |\delta_\alpha| \vee |\delta_\beta| \in \R\}, \\
            \Phi_{i, \delta} &\coloneq \{\phi_i(\cdot, \cdot, \cdot; \delta_\alpha, \delta_\beta): \mathcal{X} \times \mathcal{E} \times \mathcal{B} \to \R: |\delta_\alpha| \vee |\delta_\beta| < \delta\},
            \numberthis \label{eq:def_bigpsi_bigphi}
        \end{align*}
        where $\tilde\psi$ and $\phi_i$ are defined in \ref{eq:def_psi_phi}.
        By Lemma \ref{lem:conditional_exp}, we have 
        for all $i \in N$,
        \begin{align*}
            N(\tilde{\Psi}_i, \|\cdot\|_{Q_B,2}, 2\epsilon) \leq 
            \left( \frac{A}{\epsilon}\right)^v,
        \end{align*}
        for any finite discrete probability measure $Q_B$ on $\mathcal{B}$. We therefore have, for any finite discrete
        probability measure $Q$ on $\mathcal{X} \times \mathcal{E} \times \mathcal{B}$ and $i \in \N$,  
        \begin{align*}
            N(\Phi_i, \|\cdot\|_{Q,2}, 2 \epsilon) \leq \left( \frac{A}{\epsilon}\right)^{2 v}.
        \end{align*}
        Now we focus on a class $\Phi_{i, \delta}$ with a $0 < \delta < \infty$.
        The class is pointwise measurable.
        Notice that each function in $\Phi_{i, \delta}$ has zero mean (w.r.t to the probability distribution of $(X_{i1}, \epsilon_{i1}, B_1)$). They also have zero conditional mean given $B_t$. 
        Therefore, by the law of total variance, it holds that
        \begin{align*}
            \var\left( \phi_i(X_{i1}, \epsilon_{i1}, B_1; \delta_\alpha, \delta_\beta)\right) 
            &= \E\left[ \var\left( \phi_i(X_{i1}, \epsilon_{i1}, B_1; \delta_\alpha, \delta_\beta)\given B_1\right)\right] \\
            &= \E\left[ \var\left( \idf\{\epsilon_{i1} \leq \delta_\beta'X_{i1} + \delta_\alpha\} - \idf\{\epsilon_{i1} \leq 0\}\given B_1\right)\right] \\
            &\leq \E\left[ \E[ (\idf\{\epsilon_{i1} \leq \delta_\beta'X_{i1} + \delta_\alpha\} - \idf\{\epsilon_{i1} \leq 0\})^2\given B_1]\right] \\
            &= \E\left[ \left|( F_i(X_{i1}'\delta_\beta + \delta_\alpha \given X_{i1}, B_1) - F_i(0 \given X_{i1}, B_1) \right|\right] \\
            &\leq C_f \left(|\delta_\alpha| + \sup_{x \in X} \norm{x}|\delta_\beta|\right).
        \end{align*} 
        Therefore, we can apply Corollary 5.1 in \cite{chernozhukov2014gaussian} to $\Phi_{i, \delta}$ with $F = 2$ and 
        $\sigma^2 = C_f (1 + \sup_{x \in \mathcal{X}} \norm{x} )\delta$.
        This yields 
        \begin{align}
            \E \norm{ \sum_{t = 1}^T \phi(X_{it}, \epsilon_{it}, B_t)}_{\Phi_{i, \delta}} \leq C \left( |\log \delta| + \sqrt{T} \delta^{1/2} |\log \delta|^{1/2} \right). \label{eq:bound_expectation_sup}
        \end{align}
        where $C$ is a constant independent of $i$, $N$ and $T$.
        It follows that, for any $i \in \mathbb{N}$ and  
        $|\hat{\delta}_\alpha| \vee |\hat{\delta}_\beta| < \delta_N$, we have 
        \begin{align*}
            \E\norm{ \sum_{t = 1}^T \phi_i(X_{it}, \epsilon_{it}, B_t; \hat{\delta}_\alpha, \hat{\delta}_\beta)}
            \leq C \left( |\log \delta_N| + \sqrt{T} \delta_N^{1/2}|\log \delta_N|^{1/2} \right).
        \end{align*}
        This proves \eqref{eq:2026-0131-1355}.
    \end{proof}
\end{lemma}

\begin{lemma}[Projection error bound (ii)]\label{lem:hajek_rs_2}
Under the setup of Section \ref{sec:model_estimation},
   suppose Assumptions \ref{assump:x_bounded_support} - \ref{assump:bounded_pdf_derivative} hold, then
    \begin{align*}
      \|  \mathbb{H}_N^{(2)}(\hat{\alpha}, \hat{\beta}) - \tilde{\mathbb{H}}_N^{(2)}( \hat{\alpha}, \hat{\beta}) \|
        = o_p(T^{-1/2}).
    \end{align*}
   
    \begin{proof}
    The proof proceeds along the same lines as the second part of Lemma \ref{lem:hajek_rs_1} (from \eqref{eq:2026-0131-1355} onward) and is therefore omitted.

    \end{proof}
\end{lemma}

\begin{lemma} [Projection error bound (iii)] \label{lem:20260129}
Under the setup of Section \ref{sec:model_estimation},
    suppose Assumptions \ref{assump:x_bounded_support}-\ref{assump:bounded_pdf_derivative} hold. 
    Further assume that $(\log N)^2 / T \to 0$. Then 
    \begin{align*}
        \max_{1 \leq i \leq N} \left|\mathbb{H}_{Ni}^{(1)}(\hat{\alpha}_i, \hat{\beta}) - \tilde{\mathbb{H}}_{Ni}^{(1)}(\hat{\alpha}_i, \hat{\beta})\right| = O_p\left( \sqrt{\frac{\log N}{T}}\right).
    \end{align*}
    \begin{proof}
        Let $\phi_i(x, \epsilon, b; \delta_\alpha, \delta_\beta)$ be defined 
        as in \eqref{eq:def_psi_phi} in the proof of Lemma \ref{lem:hajek_rs_1}.
        Observe that 
        \begin{align*}
            \mathbb{H}_{Ni}^{(1)}(\hat{\alpha}_i, \hat{\beta}) - \tilde{\mathbb{H}}_{Ni}^{(1)}(\hat{\alpha}_i, \hat{\beta}) 
            = \left( \mathbb{H}_{Ni}^{(1)}(\alpha_{i0}, \beta_{i0}) - \tilde{\mathbb{H}}_{Ni}^{(1)}(\alpha_{i0}, \beta_{i0}) \right)
            + \frac{1}{T} \sum_{t = 1}^T \phi_i(X_{it}, \epsilon_{it}, B_t; \hat{\delta}_{\alpha i}, \hat{\delta}_\beta).
        \end{align*}
It then suffices to show that the two terms on the right hand side are of the order $O_p(\sqrt{T^{-1}\log N})$
        
        We first show that 
        \begin{align}
            \max_{1 \leq i \leq N} \left|\mathbb{H}_{Ni}^{(1)}(\alpha_{i0}, \beta_{i0}) - \tilde{\mathbb{H}}_{Ni}^{(1)}(\alpha_{i0}, \beta_{i0})\right| = O_p\left( \sqrt{\frac{\log N}{T}} \right) \label{eq:2026-0202-2250}
        \end{align}
        By Hoeffding's Inequality, for any $s \in \mathbb{R}$ and $i \in \N$,
        we have 
        \begin{align*}
            P\left\{ \left|\mathbb{H}_{Ni}^{(1)}(\alpha_{i0}, \beta_{i0}) - \tilde{\mathbb{H}}_{Ni}^{(1)}(\alpha_{i0}, \beta_{i0})\right| > s\right\} \leq 
            \exp(-Ts^2 / 2).
        \end{align*}
       By the union bound, we have
        \begin{align*}
            &P\left\{\max_{1 \leq i \leq N}\left|\mathbb{H}_{Ni}^{(1)}(\alpha_{i0}, \beta_{i0}) - \tilde{\mathbb{H}}_{Ni}^{(1)}(\alpha_{i0}, \beta_{i0})\right| > s\right\} \leq N \exp(-Ts^2 / 2),
              \end{align*}
              which in turn implies
                  \begin{align*}
              P\left\{\max_{1 \leq i \leq N}\left|\mathbb{H}_{Ni}^{(1)}(\alpha_{i0}, \beta_{i0}) - \tilde{\mathbb{H}}_{Ni}^{(1)}(\alpha_{i0}, \beta_{i0})\right| > \sqrt{\frac{\log N}{T}} s\right\} \leq \exp(-2s^2).
        \end{align*}
        
        Next, we show that 
        \begin{align}
            \max_{1 \leq i \leq N} \left|\frac{1}{T} \sum_{t = 1}^T \phi_i(X_{it}, \epsilon_{it}, B_t; \hat{\delta}_{\alpha i}, \hat{\delta}_\beta)\right| = o_p\left( \sqrt{\frac{\log N}{T}} \right).
        \end{align}
        The argument builds on the approach used in Step 3 of the proof of Theorem 3.2 in \citet{kato2012asymptotics} with appropriate modifications.
        For each $i \in \N$, let $\Phi_i$ and $\Phi_{i, \delta}$ be defined as
        in \eqref{eq:def_bigpsi_bigphi}.
        By the union bound and 
        $\max_{1 \leq i \leq N} |\hat{\alpha}_i - \alpha_{i0}| \vee \norm{\hat{\beta} - \beta_0} \pto 0$, it suffices to prove that 
        for any $\epsilon  > 0$, there exists $\delta > 0$ such that 
        \begin{align}
            \max_{1 \leq i \leq N} P\left\{ \norm{\sum_{t = 1}^T \phi(X_{it}, \epsilon_{it}, B_t)}_{\Phi_{i, \delta}} > (T \log N)^{1/2} \epsilon \right\} = o( 1/ N). \label{eq:2026-0203-0056}
        \end{align}
        Fix $\epsilon > 0$. For any $\delta > 0$ and $1 \leq i \leq N$, write $Z_{Ni}(\delta) = \norm{\sum_{t = 1}^T \phi(X_{it}, \epsilon_{it}, B_t)}_{\Phi_{i, \delta}}$ and let $\sigma(\delta) = C_f (1 + \sup_{x \in \mathcal{X}} \norm{x}) \delta$.
        Recall from the proof of Lemma \ref{lem:hajek_rs_1}, 
        \[ \sigma(\delta)^2 \geq \sup_{\phi \in \phi_{i, \delta}} \E\left[ \phi(X_{i1}, \epsilon_{i1}, B_1)^2 \right].\]
        Also, note that functions in $\Phi_{i, \delta}$ are centred.
        Applying Proposition B.2 in \citealt{kato2012asymptotics} (Bousquet’s version of Talagrand’s inequality) with $U = 2$, for any sequence $s_N > 0$,
        \begin{align*}
            P \left\{Z_{Ni}(\delta) \geq \E[Z_{Ni}(\delta)] + s_N\sqrt{2(T \sigma(\delta)^2 + 2U \E[Z_{Ni}(\delta)])} + \frac{s_N^2 U}{3} \right\} \leq \exp(-s_N^2).
        \end{align*}
        Take $s_N = \sqrt{2 \log N}$. Then the right hand side of the inequality becomes $(1/N)^2 = o(1 / N)$. It then suffices to prove there exists $\delta > 0$ independent of $i$, $N$ and $T$
        such that 
        \begin{align*}
            &\E[Z_{Ni}(\delta)] + s_N\sqrt{2(T \sigma(\delta)^2 + 2U \E[Z_{Ni}(\delta)])} + \frac{s_N^2 U}{3} \leq (T \log N)^{1/2} \epsilon \\
            \iff&(T \log N)^{-1/2}\left( \E[Z_{Ni}(\delta)] + s_N\sqrt{2(T \sigma(\delta)^2 + 2U \E[Z_{Ni}(\delta)])} + \frac{s_N^2 U}{3}\right) \leq \epsilon.
        \end{align*}
        for large enough $N$. 
        With the assumption that $(\log N)^2 / T \to 0$, we have
        \begin{align*}
            (T \log N)^{-1/2} \frac{s_N^2 U}{3} \longto 0.
        \end{align*}
        By \eqref{eq:bound_expectation_sup}, $\E[Z_{Ni}(\delta)] \leq C \left( |\log \delta| + \sqrt{T} \delta|\log \delta|^{1/2} \right)$. 
        It follows that there exists a choice of $\delta$ that only depends on $C$ such that 
        \begin{align*}
            (T \log N)^{-1/2}\E[Z_{Ni}(\delta)] < \frac{\epsilon}{3}, \quad (T \log N)^{-1/2} s_N\sqrt{2(T \sigma^2 + 2U \E[Z_{Ni}(\delta)])} < \frac{\epsilon}{3}.
        \end{align*}
        for large enough $N$.
        Recall that such $C$ is a constant independent of $i$, $N$ and $T$, 
        so the choice of $\delta$ is also independent of $i$, $N$ and $T$.
        We have thus proved \eqref{eq:2026-0203-0056}.
    \end{proof}
\end{lemma}
\begin{lemma}[Magnitude of projection] \label{lem:20251229_1}

Under the setup of Section \ref{sec:model_estimation}, it holds that 
    \begin{align*}
        \max_{1 \leq i \leq N} |\tilde{\mathbb{H}}^{(1)}_{Ni}(\alpha_{i0}, \beta_0)| = O_p\left( \sqrt{\frac{\log N}{T}}\right).
    \end{align*}
    \begin{proof}
        By Hoeffding's Inequality, for any $i\in\mathbb N$, $s \in \mathbb{R}$,
        \begin{align*}
            P\{|\tilde{\mathbb{H}}^{(1)}_{Ni}(\alpha_{i0}, \beta_0)| > s\} \leq \exp \left( -\frac{2 T s^2}{(\tau - (\tau - 1))^2}\right) = \exp(-2Ts^2).
        \end{align*}
       By the union bound, we have
        \begin{align*}
            &P\left\{\max_{1 \leq i \leq N}|\tilde{\mathbb{H}}^{(1)}_{Ni}(\alpha_{i0}, \beta_0)| > s\right\} \leq N \exp(-2Ts^2),
              \end{align*}
              which in turn implies
                  \begin{align*}
              P\left\{\max_{1 \leq i \leq N}|\tilde{\mathbb{H}}^{(1)}_{Ni}(\alpha_{i0}, \beta_0)| > \sqrt{\frac{\log N}{T}} s\right\} \leq \exp(-2s^2),
        \end{align*}
        which concludes the proof.
    \end{proof}
\end{lemma}

\begin{lemma}[Uniform convergence rate for $\hat \alpha_i$]\label{lem:max_alpha_rate}
Under the setup of Section \ref{sec:model_estimation},
suppose Assumptions \ref{assump:x_bounded_support}-- \ref{assump:Jacobian}  hold, then 
\begin{align*}
    \max_{1 \leq i \leq N} |\hat{\alpha}_i - \alpha_{i0}| = O_p\left( \sqrt{\frac{\log N}{T}}\right).
\end{align*}
\end{lemma}
\begin{proof}

By Lemma \ref{lemma:pdf_gc} and Assumptions \ref{assump:x_bounded_support}, \ref{assump:bounded_pdf_derivative}, \ref{assump:variance_bounded_gamma}, and \ref{assump:Jacobian}, the first term of the right hand side of 
\eqref{eq:beta_representation} is $O_p(1/\sqrt{T})$ following a variance calculation.
Together with Lemma \ref{lem:hajek_rs_1}, \ref{lem:hajek_rs_2}, and \ref{lem:convergence_Gamma}, we can further infer from  \eqref{eq:beta_representation} that 
\begin{align}
    \norm{\hat{\beta} - \beta_0} = O_p\left( \max_{1 \leq i \leq N} |\hat{\alpha}_i - \alpha_{i0}|^2 \right) + O_p\left( \frac{1}{\sqrt{T}}\right). \label{eq:beta_primitive_rate}
\end{align}
Hence, by plugging \eqref{eq:beta_primitive_rate}
into \eqref{eq:alpha_representation} and applying Lemma \ref{lem:20260129}, \ref{lem:20251229_1}, and \ref{lem:convergence_Gamma},  we obtain the desired result.
\end{proof}

\begin{lemma}[CLT for the score] \label{lem:clt} 
    Under the setup of Section \ref{sec:model_estimation}, suppose Assumption \ref{assump:variance_bounded_gamma} holds,
    then we have 
    \begin{align*}
        \sqrt{T} \left(  \tilde{\mathbb{H}}_N^{(2)}(\alpha_0, \beta_0) - \frac{1}{N} \sum_{i = 1}^N\gamma_i \tilde{\mathbb{H}}_{Ni}^{(1)}(\alpha_{i0}, \beta_0) \right) \longdto 
        \mathcal{N}(0, \Sigma).
    \end{align*}
\begin{proof}
Observe that for any fixed $N$ and $T$, it holds that
    \begin{align*}
&\sqrt{T}\left(\tilde{\mathbb{H}}_{N}^{(2)} - \frac{1}{N} \sum_{i = 1}^N \gamma_i \tilde{\mathbb{H}}_{Ni}^{(1)}(\alpha_{i0}, \beta_0)\right) \\ 
        =& \frac{1}{T} \sum_{t = 1}^T \frac{1}{N} \sum_{i = 1}^N \E[(\tau - \idf\{\epsilon_{it} \leq 0\})(X_{it} - \gamma_i) \mid B_t] \\
        =& \frac{1}{\sqrt{T}} \sum_{t = 1}^T \underbrace{\E \left[\frac{1}{N} \sum_{i = 1}^N (\tau - \idf\{\epsilon_{it} \leq 0\})(X_{it} - \gamma_i) \,\bigg|\, B_t\right]}_{=:\psi_{Nt}}
       = \frac{1}{\sqrt{T}} \sum_{t = 1}^T \psi_{Nt}.
    \end{align*}
    Since $\psi_{Nt}$ is uniformly bounded across $N$, and $\var(\psi_{Nt})$ is positive-definite 
    for large $N$, by Lindeberg-Feller CLT, the conclusion follows.
\end{proof}
\end{lemma}
\begin{lemma}[ULLN for conditional densities] \label{lemma:pdf_gc}
Under the setup of Section \ref{sec:model_estimation}, suppose that Assumption \ref{assump:bounded_pdf_derivative} holds.
Define the sequences of function classes
\begin{align*}
\mathcal F_{N1}
&= \bigl\{\, b \mapsto f_i(0 \mid B_t = b) : i=1,\ldots,N \bigr\}, \\
\mathcal F_{N2}
&= \bigl\{\, b \mapsto \E\!\bigl[f_i(0 \mid X_{it}, B_t) X_{it} \mid B_t = b\bigr] : i=1,\ldots,N \bigr\}.
\end{align*}
Then, for each $j \in \{1,2\}$, it holds that
\[
\max_{f \in \mathcal F_{Nj}}
\left\|
\frac{1}{T} \sum_{t=1}^T \bigl\{ f(B_t) - \E[f(B_1)] \bigr\}
\right\|
= O_p\left(\sqrt{\frac{\log N}{T}}\right),
\]
as $N,T \to \infty$.

\end{lemma}
\begin{proof}
Notice that for each fixed $i$, the summands are i.i.d. over $t$.
     The desired result is then direct consequence of  Theorem 2.14.1 in \cite{VanDerVaarWellner1996EmpiricalProcess} and Markov's inequality.
\end{proof}

\begin{lemma}[Consistency of the Jacobian] \label{lem:convergence_Gamma}
    Under the setup of Section \ref{sec:model_estimation}, suppose Assumptions \ref{assump:x_bounded_support}, \ref{assump:bounded_pdf_derivative} and \ref{assump:Jacobian} hold,
    then we have $\Gamma_N \longpto \Gamma$.
\end{lemma}
\begin{proof}
  This follows directly from Lemma \ref{lemma:pdf_gc} and a standard weak law of large numbers.

\end{proof}
\begin{lemma}[CLT for the score under $M$-dependence] \label{lem:clt_m}
    Under the setup of Section \ref{sec:model_estimation}, suppose Assumption \ref{assump:variance_bounded_gamma_m} holds,
    then we have 
    \begin{align*}
        \sqrt{T} \left(  \tilde{\mathbb{H}}_N^{(2)}(\alpha_0, \beta_0) - \frac{1}{N} \sum_{i = 1}^N\gamma_i \tilde{\mathbb{H}}_{Ni}^{(1)}(\alpha_{i0}, \beta_0) \right) \longdto 
        \mathcal{N}(0, \Sigma_M).
    \end{align*}
\end{lemma}
\begin{proof}
  The proof is similar to Theorem \ref{lem:clt}. We can rewrite 
  the expression as
  \begin{align*}
      \frac{1}{\sqrt{T}} \sum_{t = 1}^T \E\left[ \frac{1}{N} \sum_{i = 1}^N (\tau - \idf\{\epsilon_{it} \leq 0\}(X_{it} - \gamma_i)) \given[\bigg] B_t\right] = \frac{1}{\sqrt{T}} \sum_{t = 1}^T \psi_{Nt}.
  \end{align*}
    Note that $\psi_{Nt}$ is uniformly bounded across $N$ and that 
    \begin{align*}
        \var\left( \frac{1}{\sqrt{T}} \sum_{t = 1}^T \psi_{Nt}\right)
        = \Gamma^0_N + \sum_{l = 1}^M \left( 1 - \frac{l}{T}\right)\left( \Gamma_N^l + {\Gamma_N^l}'\right)= \Lambda + o(1).
    \end{align*}
    The result then follows by combining the standard blocking argument for $M$-dependent sequences, as in the proof of Theorem 1 of \cite{Hoeffding1948central}, with the Lindeberg--Feller theorem for independent triangular arrays.
\end{proof}



\newpage

\bibliographystyle{ecta}
\bibliography{biblio}

@article{Hoeffding1948central,
  title   = {The Central Limit Theorem for Dependent Random Variables},
  author  = {Hoeffding, Wassily and Robbins, Herbert},
  journal = {Duke Mathematical Journal},
  volume  = {15},
  number  = {3},
  pages   = {773--780},
  year    = {1948}
}

@article{MaLintonGao2021,
title = {Estimation and Inference in Semiparametric Quantile Factor Models},
journal = {Journal of Econometrics},
volume = {222},
number = {1, Part B},
pages = {295-323},
year = {2021},
note = {Annals Issue: Financial Econometrics in the Age of the Digital Economy},
author = {Shujie Ma and Oliver Linton and Jiti Gao}
}

@ARTICLE{ChernozhukovHansen06,
   AUTHOR="Victor Chernozhukov and Christian Hansen",
   TITLE="Instrumental Quantile Regression Inference for Structural and Treatment Effects Models",
   JOURNAL={Journal of Econometrics},
   VOLUME="132",
   PAGES="491--525",
   YEAR="2006"
}

@article{Chen2024, 
title={Two-Step Estimation of Quantile Panel Data Models with Interactive Fixed Effects}, 
volume={40}, 
number={2}, 
journal={Econometric Theory}, 
author={Chen, Liang}, 
year={2024}, 
pages={419–446}
}

@article{BelloniChenPadillaWang2023,
author = {Alexandre Belloni and Mingli Chen and Oscar Hernan Madrid Padilla and Zixuan (Kevin) Wang},
title = {{High-dimensional latent panel quantile regression with an application to asset pricing}},
volume = {51},
journal = {The Annals of Statistics},
number = {1},
publisher = {Institute of Mathematical Statistics},
pages = {96 -- 121},
year = {2023}
}

@article{HardingLamarchePesaran2020,
author = {Harding, Matthew and Lamarche, Carlos and Pesaran, M. Hashem},
title = {Common Correlated Effects Estimation of Heterogeneous Dynamic Panel Quantile Regression Models},
journal = {Journal of Applied Econometrics},
volume = {35},
number = {3},
pages = {294-314},
year = {2020}
}

@article{AndoBai2020,
author = {Tomohiro Ando and Jushan Bai},
title = {Quantile Co-Movement in Financial Markets: A Panel Quantile Model With Unobserved Heterogeneity},
journal = {Journal of the American Statistical Association},
volume = {115},
number = {529},
pages = {266--279},
year = {2020}
}

@article{chernozhukov2013average,
  title={Average and quantile effects in nonseparable panel models},
  author={Chernozhukov, Victor and Fern{\'a}ndez-Val, Iv{\'a}n and Hahn, Jinyong and Newey, Whitney},
  journal={Econometrica},
  volume={81},
  number={2},
  pages={535--580},
  year={2013},
  publisher={Wiley Online Library}
}

@article{fernandez2018fixed,
  title={Fixed effects estimation of large-T panel data models},
  author={Fern{\'a}ndez-Val, Iv{\'a}n and Weidner, Martin},
  journal={Annual Review of Economics},
  volume={10},
  number={1},
  pages={109--138},
  year={2018},
  publisher={Annual Reviews of Economics}
}

@article{chiang2024standard,
  title={Standard errors for two-way clustering with serially correlated time effects},
  author={Chiang, Harold D and Hansen, Bruce E and Sasaki, Yuya},
  journal={Review of Economics and Statistics},
  pages={1--40},
  year={2024},
  publisher={MIT Press}
}

@article{juodis2025shock,
  title={This shock is different: Estimation and inference in misspecified two-way fixed effects panel regressions},
  author={Juodis, Art{\=u}ras},
  journal={Econometric Theory},
  pages={1--34},
  year={2025},
  publisher={Cambridge University Press}
}

@article{besstremyannaya2019reconsideration,
  title={Reconsideration of a simple approach to quantile regression for panel data},
  author={Besstremyannaya, Galina and Golovan, Sergei},
  journal={The Econometrics Journal},
  volume={22},
  number={3},
  pages={292--308},
  year={2019},
  publisher={Oxford University Press}
}

@article{hausman2025linear,
  title={Linear Estimation of Structural and Causal Effects for Nonseparable Panel Data},
author={Chernozhukov, Victor and Deaner, Ben and Gao, Ying and Hausman, Jerry A and Newey, Whitney},
  year={2025}
}

@techreport{andrews2003cross,
  title={Cross-section Regression with Common Shocks},
  author={Andrews, Donald},
  year={2003},
  institution={Cowles Foundation for Research in Economics, Yale University}
}

@Article{GuVolgushev19,
  author    = {Jiaying Gu and Stanislav Volgushev},
  title     = {Panel Data Quantile Regression with Grouped Fixed Effects},
  journal   = {Journal of Econometrics},
  year      = {2019},
  volume    = {213},
  pages     = {68-91}
}

@Article{Koenker04,
  Title                    = {Quantile Regression for Longitudinal Data},
  Author                   = {Roger Koenker},
  Journal                  = {Journal of Multivariate Analysis},
  Year                     = {2004},
  Pages                    = {74--89},
  Volume                   = {91}
}

@ARTICLE{AbrevayaDahl08,
   AUTHOR = {Jason Abrevaya and Christian M. Dahl},
   TITLE = {The Effects of Birth Inputs on Birthweight: Evidence From Quantile Estimation on Panel Data},
   JOURNAL = {Journal of Business and Economic Statistics},
   VOLUME = {26},
   PAGES = {379-397},
   YEAR = {2008}
}

@Article{Lamarche10,
  Title                    = {Robust Penalized Quantile Regression Estimation for Panel Data},
  Author                   = {Carlos Lamarche},
  Journal                  = {Journal of Econometrics},
  Year                     = {2010},
  Pages                    = {396--408},
  Volume                   = {157}
}

@Article{MachadoSantosSilva19,
  author    = {Jos/'e A. F. Machado and J. M. C. {Santos Silva}},
  title     = {Quantiles via Moments},
  journal   = {Journal of Econometrics},
  year      = {2019},
  volume    = {213},
  pages     = {145-173}
}

@Article{ChetverikovLarsenPalmer16,
  Title = {IV Quantile Regression for Group-Level Treatments, with an Application to the Effects of Trade on the Distribution of Wages},
  Author = {Denis Chetverikov and Bradley Larsen and Christopher Palmer},
  Journal= {Econometrica},
  Year  = {2016},
  Volume= {84},
  Pages = {809–833}
}

@Article{ArellanoBonhomme16,
  Title                    = {Nonlinear Panel Data Estimation Via Quantile Regressions},
  Author                   = {Manuel Arellano and  Stephane Bonhomme},
  Journal                  = {Econometrics Journal},
  Year                     = {2016},
  Pages                    = {C61--C94},
  Volume                   = {19}
}

@article{ZhangWangLianLi2023,
author = {Xiaoyu Zhang and Di Wang and Heng Lian and Guodong Li},
title = {Nonparametric Quantile Regression for Homogeneity Pursuit in Panel Data Models},
journal = {Journal of Business \& Economic Statistics},
volume = {41},
number = {4},
pages = {1238--1250},
year = {2023}
}

@article{KimYang11,
title = {Semiparametric Approach to a Random Effects Quantile Regression Model},
author = {Mi-Ok Kim and Yunwen Yang},
journal = {Journal of the American Statistical Association},
year = {2011},
pages = {1405--1417},
volume = {106}
}

@Article{GrahamHahnPoirierPowell2015,
  Title                    = {A Quantile Correlated Random Coefficients Panel Data Model},
  Author                   = {Bryan S. Graham and Jinyong Hahn and Alexandre Poirier and James L. Powell},
  Journal                  = {Journal of Econometrics},
  Volume		      ={206},
  Pages                  ={305-335},
  Year                     = {2018}
}

@incollection{ArellanoHahn07,
	Author = {Manuel Arellano and Jinyong Hahn},
	Booktitle = {Economics and Econometrics},
	Editor = {R Blundell and W K Newey and T Persson},
	Pages = {381--409},
	Publisher = {Cambridge University Press},
	Title = {Understanding Bias in Nonlinear Panel Models: Some Recent Developments},
	Year = {2007}}

@article{ArellanoBonhomme11,
	Author = {Manuel Arellano and Stephane Bonhomme},
	Journal = {Annual Review of Economics},
	Pages = {395--424},
	Title = {Nonlinear Panel Data Analysis},
	Volume = {3},
	Year = {2011}}

@article{FernandezValWeidner18,
	Author = {Ivan Fern\'{a}ndez-Val and Martin Weidner},
	Title = {Fixed Effect Estimation of Large-T Panel Data Models},
    Journal={Annual Review of Economics},
    Volume={10},
    Pages={109-138},
	Year = {2018}}

@article{berry1995automobile,
  title={Automobile prices in market equilibrium},
  author={Berry, Steven and Levinsohn, James and Pakes, Ariel},
  journal={Econometrica},
  volume={63},
  number={4},
  pages={841--890},
  year={1995},
  publisher={Blackwell Publishing Inc.}
}

@book{moulin1991axioms,
  title={Axioms of cooperative decision making},
  author={Moulin, Herv{\'e}},
  number={15},
  year={1991},
  publisher={Cambridge university press}
}

@article{athey2001investment,
  title={Investment and market dominance},
  author={Athey, Susan and Schmutzler, Armin},
  journal={RAND Journal of economics},
  pages={1--26},
  year={2001},
  publisher={JSTOR}
}

@book{kallenberg2021foundation,
  author    = {Kallenberg, Olav},
  title     = {Foundations of Modern Probability},
  publisher = {Springer},
  year      = {2021}
}

@book{kallenberg2005probabilistic,
  title={Probabilistic symmetries and invariance principles},
  author={Kallenberg, Olav},
  year={2005},
  publisher={Springer}
}

@article{ghosal2000testing,
  title={Testing monotonicity of regression},
  author={Ghosal, Subhashis and Sen, Arusharka and van der Vaart, Aad W},
  journal={Annals of Statistics},
  pages={1054--1082},
  year={2000},
  publisher={JSTOR}
}

@article{fernandez2005bias,
  title={Bias correction in panel data models with individual specific parameters},
  author={Fern{\'a}ndez-Val, Iv{\'a}n},
  journal={Available at SSRN 869104},
  year={2005}
}

@book{demetrescu2023tests,
  title={Tests of no cross-sectional error dependence in panel quantile regressions},
  author={Demetrescu, Matei and Hosseinkouchack, Mehdi and Rodrigues, Paulo MM},
  number={1041},
  year={2023},
  publisher={Ruhr Economic Papers}
}

@article{andrews2005cross,
  title={Cross-section regression with common shocks},
  author={Andrews, Donald WK},
  journal={Econometrica},
  volume={73},
  number={5},
  pages={1551--1585},
  year={2005},
  publisher={Wiley Online Library}
}

@article{chernozhukov2014gaussian,
  title={Gaussian approximation of suprema of empirical processes},
  author={Chernozhukov, Victor and Chetverikov, Denis and Kato, Kengo},
  journal={The Annals of Statistics},
  volume={42},
  number={4},
  pages={1564},
  year={2014},
  publisher={Institute of Mathematical Statistics}
}

@article{athey2025identification,
  title={Identification of average treatment effects in nonparametric panel models},
  author={Athey, Susan and Imbens, Guido},
  journal={arXiv preprint arXiv:2503.19873},
  year={2025}
}

@article{chen2021quantile,
  title={Quantile factor models},
  author={Chen, Liang and Dolado, Juan J and Gonzalo, Jes{\'u}s},
  journal={Econometrica},
  volume={89},
  number={2},
  pages={875--910},
  year={2021},
  publisher={Wiley Online Library}
}

@article{galvao2016smoothed,
  title={Smoothed quantile regression for panel data},
  author={Galvao, Antonio F and Kato, Kengo},
  journal={Journal of Econometrics},
  volume={193},
  number={1},
  pages={92--112},
  year={2016},
  publisher={Elsevier}
}

@article{koenker1978regression,
  title={Regression quantiles},
  author={Koenker, Roger and Bassett, Gilbert},
  journal={Econometrica: journal of the Econometric Society},
  pages={33--50},
  year={1978},
  publisher={JSTOR}
}

@article{canay2011simple,
  title={A simple approach to quantile regression for panel data},
  author={Canay, Ivan A},
  journal={Econometrics Journal},
  volume={14},
  number={3},
  pages={368--386},
  year={2011},
  publisher={Oxford University Press Oxford, UK}
}

@article{galvao2020unbiased,
  title={On the unbiased asymptotic normality of quantile regression with fixed effects},
  author={Galvao, Antonio F and Gu, Jiaying and Volgushev, Stanislav},
  journal={Journal of Econometrics},
  volume={218},
  number={1},
  pages={178--215},
  year={2020},
  publisher={Elsevier}
}

@article{fama1973risk,
  title={Risk, return, and equilibrium: Empirical tests},
  author={Fama, Eugene F and MacBeth, James D},
  journal={Journal of Political Economy},
  volume={81},
  number={3},
  pages={607--636},
  year={1973},
  publisher={The University of Chicago Press}
}

@article{driscoll1998consistent,
  title={Consistent covariance matrix estimation with spatially dependent panel data},
  author={Driscoll, John C and Kraay, Aart C},
  journal={Review of Economics and Statistics},
  volume={80},
  number={4},
  pages={549--560},
  year={1998},
  publisher={MIT Press 238 Main St., Suite 500, Cambridge, MA 02142-1046, USA journals~…}
}

@article{petersen2008estimating,
  title={Estimating standard errors in finance panel data sets: Comparing approaches},
  author={Petersen, Mitchell A},
  journal={The Review of Financial Studies},
  volume={22},
  number={1},
  pages={435--480},
  year={2008},
  publisher={Society for Financial Studies}
}

@book{VanDerVaarWellner1996EmpiricalProcess,
publisher={Springer New York, NY},
series={Springer Series in Statistics},
author={van der Vaart,A. W. and Jon A. Wellner},
title = {Weak Convergence and Empirical Processes},
year={1996}
}

@article{graham2024sparse,
  title={Sparse network asymptotics for logistic regression under possible misspecification},
  author={Graham, Bryan S},
  journal={Econometrica},
  volume={92},
  number={6},
  pages={1837--1868},
  year={2024},
  publisher={Wiley Online Library}
}

@article{bickel2011method,
  title={The method of moments and degree distributions for network models},
  author={Bickel, Peter J and Chen, Aiyou and Levina, Elizaveta},
  journal={Annals of Statistics},
  volume={39},
  number={5},
  pages={2280--2301},
  year={2011}
}

@article{kato2012asymptotics,
  title={Asymptotics for panel quantile regression models with individual effects},
  author={Kato, Kengo and Galvao, Antonio F and Montes-Rojas, Gabriel V},
  journal={Journal of Econometrics},
  volume={170},
  number={1},
  pages={76--91},
  year={2012},
  publisher={Elsevier}
}

@article{davezies2021empirical,
  title={Empirical process results for exchangeable arrays},
  author={Davezies, Laurent and D’Haultf{\oe}uille, Xavier and Guyonvarch, Yannick},
  journal={The Annals of Statistics},
  volume={49},
  number={2},
  pages={845--862},
  year={2021},
  publisher={JSTOR}
}

@article{menzel2021bootstrap,
  title={Bootstrap with cluster-dependence in two or more dimensions},
  author={Menzel, Konrad},
  journal={Econometrica},
  volume={89},
  number={5},
  pages={2143--2188},
  year={2021},
  publisher={Wiley Online Library}
}
\end{document}